\documentclass[11pt]{article}
\usepackage{graphicx}
\usepackage{pgfplots}
\usepackage{subcaption}
\usepackage[font=scriptsize]{caption} 

\usepackage[colorlinks=true,citecolor=blue]{hyperref}
\usepackage[round, authoryear]{natbib}
\usepackage{graphicx}
\usepackage{amsfonts}
\usepackage{amsmath}
\usepackage{amssymb}
\usepackage{url}
\usepackage{fancyhdr}
\usepackage{indentfirst}

\usepackage{enumerate}
\usepackage{titlesec}
\usepackage{amsthm}
\usepackage{dsfont}
\usepackage[misc]{ifsym}

\theoremstyle{definition}
\newtheorem{example}{Example}[section]

\newtheorem{definition}{Definition}[section]
\newtheorem{assumption}{Assumption}[section]

\newtheorem{problem}{Problem}[section]

\theoremstyle{plain}
\newtheorem{theorem}{Theorem}[section]
\newtheorem{lemma}{Lemma}[section]
\newtheorem{proposition}{Proposition}[section]
\newtheorem{corollary}{Corollary}[section]

\theoremstyle{remark}
\newtheorem{remark}{Remark}[section]

\numberwithin{equation}{section}

\usepackage{cases}

\theoremstyle{definition}

\def\d{\,\mathrm{d}}

\usepackage[onehalfspacing]{setspace}

\usepackage{bm}
\usepackage{tikz-qtree}
\usepackage{tikz}

\pgfplotsset{compat=1.18}

\title{Optimal Insurance Menu Design under the Expected-Value Premium Principle}
\author{Xia Han\thanks{School of Mathematical Sciences, LPMC and AAIS, Nankai University, China.  \Letter~{\url{xiahan@nankai.edu.cn}}}  \and Bin Li\thanks{Department of Statistics and Actuarial Science, University of Waterloo,  Canada. \Letter~{\url{bin.li@uwaterloo.ca}}}}

\date{\today}

\begin{document}
	\maketitle

\begin{abstract}
This paper studies optimal insurance design under asymmetric information in a Stackelberg framework, where a monopolistic insurer faces uncertainty about both the insured’s risk attitude, captured by a risk-aversion parameter, and the insured’s risk type, characterized by the loss distribution. In particular, when the risk type is unobservable, we allow the risk-aversion parameter to depend on the risk type.
We construct a menu of contracts that maximizes the mean–variance utilities of both parties under the expected-value premium principle, subject to a truth-telling constraint that ensures the truthful revelation of private information.
We show that when risk attitude is private information, the optimal coverage takes the form of excess-of-loss insurance with linear pricing in terms of the risk loading (defined as the premium minus the expected loss), designed to screen risk preferences. In contrast, when risk type is unobserved, we restrict the coverage function to an excess-of-loss form and derive an ordinary differential equation that characterizes the optimal risk loading. Under mild conditions, we establish the existence and uniqueness of the solution. The results show that equilibrium contracts exhibit nonlinear pricing with decreasing risk loadings, implying that higher-risk individuals face lower risk loadings in order to induce self-selection.
Finally, numerical illustrations demonstrate how parameter values and the distributions of unobserved heterogeneity affect the structure of optimal contracts and the resulting pricing schedule.

\end{abstract}

\section{Introduction}

Insurance markets are inherently characterized by information asymmetry between insurers and policyholders. Individuals typically possess private information about their underlying risks, such as lifestyle choices in health insurance, driving habits in automobile insurance, or managerial commitment in business interruption insurance. Although the increasing availability of observable data---including driving records, claims histories, and credit scores---allows insurers to approximate individual risk profiles, these proxies remain imperfect and substantial informational frictions persist.

The literature on optimal insurance under asymmetric information originates from the seminal contributions of \citet{rothschild1976equilibrium} for competitive markets and \citet{stiglitz1977monopoly} for monopoly markets. Subsequent studies have extended these foundations in various directions, including competitive market models (e.g., \citet{wambach2000introducing, garcia2021information, farinha2023risk}) and monopoly market models (e.g., \citet{chade2012optimal, hendren2013private, gershkov2023optimal}). This literature examines how asymmetric information shapes insurance contract design through mechanisms such as adverse selection, thereby providing fundamental insights into the functioning of insurance markets.
Beyond risk types, customers' risk preferences, including their degree of risk aversion, are also largely unobservable to insurers. As a result, insurers face uncertainty regarding both the objective characteristics of risk and the subjective attitudes toward risk. This dual informational asymmetry constitutes a fundamental challenge for the design of insurance contracts. The importance of asymmetric information in economic analysis is underscored by the 2001 Nobel Prize in Economics awarded to Akerlof, Spence, and Stiglitz for their pioneering contributions to the study of markets with asymmetric information.

One approach to addressing information asymmetry is through ambiguity modeling. Within this framework, decision makers consider a set of plausible beliefs (i.e., prior probability measures) and make decisions according to criteria such as the maximin rule, weighted combinations of worst- and best-case scenarios, or smooth aggregations over multiple beliefs. Existing studies, including \citet{gollier2014optimal}, \citet{li2016alpha}, and \citet{guan2022equilibrium}, investigate optimal contract design under such ambiguity assumptions. However, this approach typically leads to a single contract and therefore cannot effectively elicit customers' private information.

A more effective approach is to offer a menu of contracts satisfying both incentive compatibility and individual rationality constraints, thereby inducing customers with heterogeneous characteristics to reveal their private information through self-selection. In recent years, this mechanism has been widely studied in the design of insurance and reinsurance contracts under asymmetric information. For example, \citet{cheung2019reinsurance} derive optimal menus of reinsurance contracts under excess-of-loss or proportional arrangements when insurers have two types of hidden characteristics. \citet{cheung2020concave} extend the analysis to convex distortion risk measures, while \citet{liang2022revisiting} further generalize the framework by relaxing restrictions on the form of reinsurance strategies. More recently, \citet{cheung2025optimal} study continuous-type asymmetric information and provide a general characterization of optimal contract menus. Other studies consider asymmetric information regarding risk attitudes. For instance, \citet{gershkov2023optimal} extend the classical monopoly model of \citet{stiglitz1977monopoly} by representing risk-averse preferences through dual utility functions (\citet{yaari1987dual}). \citet{boonen2021optimal} analyze optimal reinsurance contracts when customers evaluate risk using distortion risk measures but the insurer does not observe the distortion function. Similarly, \citet{ghossoub2025optimal} examine contract design between a risk-neutral insurer and consumers with dual utility preferences.

Despite these advances, several limitations remain in the existing literature. First, the aforementioned work primarily focuses on decision-making under monopoly, whereas real insurance markets lie somewhere between competition and monopoly. Second, most models rely on single-period frameworks, making it difficult to capture the dynamic nature of insurance relationships. In addition, many studies assume discrete risk types or focus on a single dimension of private information. These simplifications limit the ability of theoretical models to capture the complexity of real insurance markets.  To address some of these limitations, \citet{han2026optimal} develop a continuous-time framework for optimal insurance contract design based on a Stackelberg equilibrium under the mean--variance (MV) decision criterion and the variance premium principle. In their framework, explicit optimal contract menus satisfying the truth-telling constraint are derived when either the customer's loss distribution or risk attitude is uncertain. A key finding is that when the insurer seeks to elicit information about a customer's risk type, the risk loading is set lower for high-risk individuals to induce truthful revelation. In contrast, when uncertainty concerns only risk attitudes, no such adjustment in the risk loading occurs.

Motivated by this finding, the present paper extends the framework of \citet{han2026optimal} by replacing the variance premium principle with the expected-value premium principle. Both principles are classical and widely used in the insurance literature, making it natural to investigate how the choice of premium principle affects optimal contract design under asymmetric information. Within a Stackelberg equilibrium framework, we characterize the optimal menu of contracts under mean--variance preferences. Each contract satisfies a truth-telling constraint that induces customers to reveal their private information by selecting the contract designed for their type.

In this paper, we first consider the case in which the insurer faces uncertainty regarding the customer’s risk attitude. Under the truth-telling constraint, the optimal coverage takes an excess-of-loss form and the associated risk loading remains constant across different levels of risk aversion.
We then study the case in which the insurer observes the customer’s risk attitude but faces uncertainty about the customer’s risk type. In contrast to \citet{han2026optimal}, we allow the risk-aversion parameter to depend on the risk type. While the variance premium principle leads to an analytically tractable ordinary differential equation (ODE) for proportional coverage, the expected-value premium principle introduces substantially greater analytical complexity. Although the customer’s optimal strategy can be expressed explicitly, its complexity makes a direct analytical verification of the  truth-telling constraint intractable. To obtain tractable results, we focus on excess-of-loss contracts, which are known to be optimal in classical mean--variance models without information asymmetry under the expected-value premium principle. Even with this restriction, the ODE in our setting does not yield an explicit analytical solution. Instead, we employ a fixed-point argument to establish the existence and uniqueness of the optimal contract menu under mild conditions. We further analyze the monotonicity properties of the resulting equilibrium contracts.

Our results show that the equilibrium pricing structure becomes nonlinear when risk types are private information, in contrast to the linear structure that arises when only risk attitudes are unobservable. Specifically, higher-risk customers receive lower risk loadings in order to induce them to select contracts with higher coverage levels. This mechanism reveals their true risk types through self-selection. The intuition is that risk type directly affects both claim costs and insurer profits, making it more valuable information for the insurer than risk preferences. Consequently, insurers strategically adjust pricing to elicit this information.

Numerical illustrations further highlight how the distributions of risk types and risk attitudes shape the optimal contract menu, revealing several insights absent in \citet{han2026optimal}. We observe that the overall level of risk loading is sensitive to both the dispersion and the distributional features of risk types, as well as the risk distribution itself. In particular, equilibrium risk loadings may increase or decrease overall as uncertainty rises, and the loading under full information may be either higher or lower than that under uncertainty. In the absence of information asymmetry, the competitive model implies zero expected profit for insurers, while the monopoly model suggests no welfare gains for customers in insurance purchase decisions. Our Stackelberg framework sits between these two extremes. Consequently, the equilibrium risk loading under asymmetric information falls between the competitive and monopoly benchmarks, and either a higher or lower loading relative to the full information case is possible depending on parameter configurations.

Furthermore, while we allow risk aversion to be a function of risk type and prove that the monotonicity of the risk loading is invariant to the specification of the risk aversion function, we still find that decreasing risk aversion yields the lowest overall risk loading, increasing risk aversion yields the highest, and constant risk aversion lies in between. This result is intuitive: individuals with higher risk aversion have a stronger intrinsic demand for insurance, which reduces the insurer's need to offer discounts. Additionally, deductible choices and premiums respond significantly to the functional form of risk aversion. Higher risk aversion strengthens insurance demand and leads to lower deductibles, whereas decreasing risk aversion may generate non monotonic coverage patterns when interacting with changes in underlying risk exposure.

The remainder of the paper is organized as follows. Section \ref{sec:2} reviews the mean--variance Stackelberg framework. Section \ref{sec:3} studies the case in which the insurer faces uncertainty regarding the customer's risk attitude. Section \ref{sec:4} analyzes the case in which the insurer faces uncertainty regarding the customer's risk type. Section \ref{sec:5} presents numerical illustrations. Section \ref{sec:6} concludes the paper.
 
\section{Stackelberg equilibrium}\label{sec:2}

Let $
(\Omega, \mathcal{F}, \mathbb{F}=\{\mathcal{F}_t\}_{t\in[0,T]}, \mathbb{P}) $
be a filtered probability space satisfying the usual conditions of completeness and right-continuity, where $T>0$ denotes a finite time horizon. On this space, let $N=\{N(t)\}_{t\in[0,T]}$ be a homogeneous Poisson process with constant intensity $\lambda>0$, and let $\{Y_i\}_{i\in\mathbb{N}^+}$ be a sequence of iid  positive random variables, independent of $N$, with common cumulative distribution function $F$. Write $Y \stackrel{\mathrm d}{=} Y_1$  for a generic random variable with the distribution $F$, and assume $\mathbb{E}[Y]<\infty$ and $\mathbb{E}[Y^2]<\infty$.  

We begin by recalling the classical mean-variance  (MV) Stackelberg game framework; see, e.g., \cite{chen2019stochastic}. We assume an insurance market involving two representative parties: a customer (agent) who is exposed to random losses over time, and an insurer (principal) who provides indemnification in exchange for premium payments.

In the absence of insurance, the customer's aggregate loss process follows the classical Cramér–Lundberg risk model
$$
\sum_{i=1}^{N(t)} Y_i = \int_0^t \int_0^\infty y \, N(\mathrm{d}s, \mathrm{d}y), \quad t\in[0,T],
$$
where $N(t)$ denotes the number of claims during $[0,t]$, and $Y_i$ represents the size of the $i$th claim. Here, $N(\mathrm{d}s, \mathrm{d}y)$ is a Poisson random measure with compensator $\nu(\mathrm{d}y)\,\mathrm{d}t = \lambda\,\mathrm{d}F(y)\,\mathrm{d}t$, and its compensated version is
$$
\tilde{N}(\mathrm{d}t, \mathrm{d}y) = N(\mathrm{d}t, \mathrm{d}y) - \nu(\mathrm{d}y)\,\mathrm{d}t.
$$

The customer transfers part of the loss to the insurer through a loss coverage function
$
l(t,y):[0,T]\times\mathbb{R}_{+}\rightarrow\mathbb{R},
$
where \(y\) denotes the size of a loss occurring at time \(t\), and 
the insurer determines a continuously charged premium rate, denoted $p(l)$. Under this arrangement, the customer’s surplus process evolves as
\begin{equation*} \label{eq:X_C}
X_C(t,l) = x_C - p(l)t - \int_0^t \int_0^\infty (y - l(s,y)) N(\mathrm{d}s, \mathrm{d}y),
\end{equation*}
with initial surplus $X_C(0) = x_C > 0$, and  the insurer’s wealth process is given by
\begin{equation*} \label{eq:X_I}
X_I(t,l) = x_I + p(l)t - \int_0^t \int_0^\infty l(s,y)  N(\mathrm{d}s, \mathrm{d}y),
\end{equation*}
with $X_I(0) = x_I > 0$. In $X_{C}$ and
$X_{I}$, subscript $``C"$ stands
for customer and $I$ stands for insurer.

Both the customer and the insurer evaluate their outcomes using a MV criterion. Within the Stackelberg framework, the insurer acts as the leader by first announcing a premium schedule $p(l)$, to which the customer, as the follower, responds by selecting an optimal retention strategy.  
For any given premium schedule $p(l)$, the customer’s optimization problem is given by
$$
\sup_{l \in \mathbb{A}_l} \left\{ \mathbb{E}[X_C(T,l)] - \frac{\gamma}{2} \mathrm{Var}[X_C(T,l)] \right\},
$$
where $\gamma > 0$ denotes the customer’s risk aversion coefficient,\footnote{Later, we will treat $\gamma$ as a variable; for notational simplicity, we omit the subscript $C$ here.} and $\mathbb{A}_l$ is the admissible set of the loss coverage.

Anticipating the customer’s optimal response $l^*(p)$, the insurer solves
$$
\sup_{p \in \mathbb{A}_p} \left\{ \mathbb{E}[X_I(T, l^*(p))] - \frac{\gamma_I}{2} \mathrm{Var}[X_I(T, l^*(p))] \right\},
$$
where $\gamma_I > 0$ is the insurer’s risk aversion coefficient, and $\mathbb{A}_p$ denotes the admissible set of premium strategies.


A key assumption in the classical Stackelberg framework is that the insurer has complete knowledge of the customer’s private information. As noted in the introduction, this assumption is rarely satisfied in real-world contexts. To address this limitation, we instead consider a setting in which the customer’s risk attitude, or risk type, constitutes private information that is unobservable to the insurer.

\section{Risk-attitude uncertainty}\label{sec:3}
\subsection{Modelling}

In this section, we first consider the setting where the insurer faces uncertainty about the customer’s true risk attitude,  as captured by the risk aversion parameter $\gamma$. From the insurer’s perspective,  $\gamma$ is modeled as an independent random variable with cumulative distribution function $G(\gamma)$ supported on the compact interval $[\gamma_L, \gamma_H] \subset [0, \infty)$.

Let $\tilde{\gamma}$ denote the risk attitude parameter that the customer strategically reports when selecting an insurance contract. Based on this reported parameter, the customer chooses a loss coverage $l \in \mathcal{A}_l$. The insurer, observing only the distribution of $\gamma$, sets a continuously paid premium according to the expected-value premium principle, given by
\begin{equation}\label{eq:ex}
p(l; \tilde{\gamma}) = \bigl(1 + \xi(\tilde{\gamma})\bigr) \lambda \int_0^\infty l(t,y) f(y) \,\mathrm{d}y,
\end{equation}
where $\xi(\tilde{\gamma}) \geq 0$ is the risk loading factor corresponding to risk aversion $\tilde{\gamma}$. The set of admissible loading functions is
$$
\mathcal{A}_\xi = \left\{ \xi : \xi(\tilde{\gamma}) \geq 0, \quad \forall \tilde{\gamma} \in [\gamma_L, \gamma_H] \right\}.
$$

Under this framework, the customer’s surplus process is given by
\begin{equation}\label{XC gamma}
X^\gamma_C(t, l; \tilde{\gamma}) = x_C - p(l; \tilde{\gamma}) t - \int_0^t \int_0^\infty (y - l(s,y)) N(\mathrm{d}s, \mathrm{d}y),
\end{equation}
whereas the insurer’s wealth process evolves as
\begin{equation}\label{XI gamma}
X^\gamma_I(t, l; \tilde{\gamma}) = p(l; \tilde{\gamma}) t - \int_0^t \int_0^\infty l(s,y) N(\mathrm{d}s, \mathrm{d}y).
\end{equation}
Here, the superscript $``\gamma"$
 represents the customer's true
risk attitude, while the $``\tilde
{\gamma}"$  after the semicolon  represents his chosen contract type.

We next formulate the respective optimization problems. For a customer with true risk aversion $\gamma$ selecting a contract designed for perceived risk aversion $\tilde{\gamma}$, the objective is to maximize the MV criterion of the surplus by choosing an optimal loss coverage $l$.

\begin{problem}[Customer’s problem]\label{prob PH gamma}
Define the value function for a customer with true risk aversion $\gamma$ choosing a contract indexed by $\tilde{\gamma}$ as
\begin{equation}\label{eq:V_c1}
V_C^\gamma(\tilde{\gamma}) = \sup_{l \in \mathcal{A}_l} \left\{ \mathbb{E}[X^\gamma_C(T, l; \tilde{\gamma})] - \frac{\gamma}{2} \mathrm{Var}[X^\gamma_C(T, l; \tilde{\gamma})] \right\}.
\end{equation}
The corresponding optimal insurance strategy is denoted by $\hat{l}^\gamma(\tilde{\gamma}) := \{\hat{l}^\gamma(t,y; \tilde{\gamma})\}_{y>0}$.
\end{problem}

Note that  Problem \ref{prob PH gamma} automatically ensures the satisfaction of the customer’s {participation constraint}. That is, the customer’s optimal value function is always at least as high as the value they would obtain by purchasing no insurance. This follows directly from the fact that adopting the strategy with a zero loss coverage function, $l \equiv 0$, constitutes an admissible strategy for the customer.

After solving the customer’s problem, the insurer’s task is to design a contract menu that satisfies the truth-telling constraint (\ref{TL gamma}), ensuring each customer selects the contract corresponding to their true risk attitude parameter.  Subject to this condition, the insurer maximizes the distribution-weighted MV utility across all customer types, thus achieving both truthful self-selection and optimal contract performance.

\begin{problem}\label{prob IN gamma}
The insurer’s problem is given by$$
\sup_{\xi(\cdot) \in \mathcal{A}_\xi} \int_{\gamma_L}^{\gamma_H} \left\{ \mathbb{E}\left[ X^\gamma_I\bigl(T, \hat{l}^\gamma(\gamma); \gamma \bigr) \right] - \frac{\gamma_I}{2} \mathrm{Var}\left[ X^\gamma_I\bigl(T, \hat{l}^\gamma(\gamma); \gamma \bigr) \right] \right\} \mathrm{d}G(\gamma)
$$
subject to the truth-telling constraint
\begin{equation}\label{TL gamma}
\gamma \in \mathrm{argmax}_{\tilde{\gamma} \in [\gamma_L, \gamma_H]} V_C^\gamma(\tilde{\gamma}).
\end{equation}
\end{problem}
The truth-telling constraint \eqref{TL gamma} ensures that each customer maximizes the MV utility by selecting the contract corresponding to his real risk aversion parameter $\gamma$. In other words, no customer can improve the outcome by choosing a contract intended for a different risk  attitude parameter $\tilde{\gamma}$. 

\subsection{Verification theorem}\label{sec:vt}
Problems \ref{prob PH gamma} and \ref{prob IN gamma} describe the dynamic MV  optimization problem, which is well known to suffer from the issue of time inconsistency. To address this, our paper adopts the widely used equilibrium framework, treating the decision-making process of a customer as a non-cooperative game played with their future selves. 
For the sake of simplicity in discussion, we continue to refer to it as an optimal strategy, while keeping in mind that optimality here is meant in the equilibrium sense.
 
 For $(t,x)\in
\lbrack0,T]\times%
\mathbb{R}
$, we define%
$$
J^{\gamma}(t,x,l;\tilde{\gamma})=\mathbb {E}\left[  \left.  X^\gamma_{C}%
(T,l;\tilde{\gamma})\right\vert X^\gamma_{C}(t,l;\tilde{\gamma})=x\right]
-\frac{\gamma}{2}\mathrm{Var}\left[  \left.  X^\gamma_{C}(T,l;\tilde{\gamma
})\right\vert X^\gamma_{C}(t,l;\tilde{\gamma})=x\right]  .
$$

\begin{definition}
\label{def equilibrium0} Given $(t,x)\in\lbrack0,T]\times%
\mathbb{R}
$ and $\hat{l}\in\mathbb{A}_{l}$, we define a perturbed strategy
$l^{\varepsilon}$ as
\begin{equation*}
l^{\varepsilon}(s,y)=\left\{
\begin{array}
[c]{ll}%
\bar{l}(y), & s\in\lbrack t,t+\varepsilon),\text{ }y\in%
\mathbb{R}
\text{,}\\
\hat{l}(s,y), & s\in\lbrack t+\varepsilon,T],\text{ }y\in%
\mathbb{R}
\text{,}%
\end{array}
\right.  \label{l eps}%
\end{equation*}
where $\bar{l}:%
\mathbb{R}
\rightarrow%
\mathbb{R}
$ and $\varepsilon>0$. Suppose that, for any given $(t,x)\in\lbrack0,T]\times%
\mathbb{R}
$, and any such perturbed strategy $l^{\varepsilon}\in\mathbb{A}_{l}$, we
have
$$
\liminf_{\varepsilon\rightarrow0}{\frac{J^{\gamma}(t,x,\hat{l};\tilde{\gamma
})-J^{\gamma}(t,x,l^{\varepsilon};\tilde{\gamma})}{\varepsilon}}\geq0.
$$
Then, $\hat{l}$ is called an equilibrium strategy for the dynamic MV problem
(\ref{eq:V_c1}).
\end{definition}

For any $l\in\mathbb{A}_{l}$ and $\varphi(t,x)\in C^{1,2}([0,T]\times
\mathbb{R})$, i.e., first-order continuously differentiable in $t$ and
second-order continuously differentiable in $x$, we define the
integro-differential operator%
\begin{equation}\begin{aligned}
\mathcal {L}_{l} \phi(t,x)  &  ={\frac{\phi(t,x)}{\partial
t}-\frac{\phi(t,x)}{\partial x}}\lambda\int_{0}^{\infty}\left(
1+\xi(\tilde{\gamma}))l(t,y)\right)  \mathrm{d}F(y)\\
&  +\lambda\int_{0}^{\infty}\left[  \phi(t,x-y+l(t,y))-\phi(t,x)\right]
\mathrm{d}F(y).\label{operator0}
\end{aligned}\end{equation}

We next present the extended HJB system that characterizes the value function 
$V$ and the corresponding equilibrium strategy in Theorem~\ref{veri_thm}. The proof of this result follows standard arguments; see e.g., \cite{bjork2010general} for details.

\begin{theorem}Let the integro-differential operator $\mathcal{L}_l$ be defined as in \eqref{operator0}. Suppose there exist $V(t,x), g(t,x)\in C^{1,2}([0,T])\times\mathbb R$ satisfy the following  conditions:
\begin{itemize} 
\item[(1)] For any $(t,x)\in[0,T]\times\mathbb R$, \begin{equation}
\sup_{l\in \mathbb{A}_l
}\left\{  \mathcal{L}_{l}V(t,x)-\frac{1}{2}\gamma\mathcal{L}_{l}%
g^{2}(t,x)+\gamma g(t,x)\mathcal{L}_{l}g(t,x)\right\}  =0 \label{HJB 0},%
\end{equation}
and $\hat{l}$ denotes the optimal value to achieve the supremum in $l$. 
\item[(2)] For any $(t, x) \in[0, T] \times \mathbb{R}$,
\begin{equation*}
\left\{\begin{array}{l}
\mathcal{L}_{\hat{l}}g(t,x)=0,  \\
V(T,x)=g(T,x)=x.
\end{array}\right.\label{equiva}
\end{equation*}
\item[(3)] For any $(t, x) \in[0, T] \times \mathbb{R}$, $\hat{l}(t)$, $\mathcal{L}_{\hat{l}}V(t,x)$ and $\mathcal{L}_{\hat{l}}g^2(t, x)$ are  deterministic functions of $t$ and independent of $x$.
\end{itemize}
Then $\hat{l}$  is the equilibrium strategy  and $V(t,x)=J^{\gamma}(t,x,\hat{l};\tilde{\gamma})$ is the equilibrium value function to the dynamic MV problem \eqref{eq:V_c1}. Besides, $g(t,x)=\mathbb{E}%
\left[  \left.  X^\gamma_{C}(T,\hat{l};\tilde{\gamma})\right\vert
X^\gamma_{C}(t,\hat{l};\tilde{\gamma})=x\right]$. \label{veri_thm}
\end{theorem}

\subsection{Solution of the optimal menu}\label{sec:3.3}

In this section, we determine the optimal menu of insurance contracts through a four-step procedure:
\begin{itemize}
\item[(i)] Given a contract, the customer chooses the optimal coverage level based on the reported risk attitude and the corresponding risk loading offered by the insurer.

\item[(ii)] To elicit private information, each contract must satisfy a truth-telling constraint, ensuring that the customer maximizes utility by reporting their true type. This condition yields a first-order restriction on the insurer’s pricing strategy.

\item[(iii)] Using the restriction obtained in step (ii), we characterize the insurer’s optimal menu of contracts, consisting of premium–coverage pairs that maximize the insurer’s objective subject to the truth-telling constraint.

\item[(iv)] Finally, we verify that the proposed menu indeed satisfies the truth-telling constraint, since the first-order condition provides only a necessary condition for truthfulness.
\end{itemize}

We begin by solving the customer’s problem under the premium function specified in \eqref{eq:ex}. The optimal coverage choice and the corresponding value function are characterized in Theorem~\ref{thm PH gamma}.
\begin{theorem}
\label{thm PH gamma}An optimal insurance coverage to the customer's problem
(\ref{eq:V_c1}) is given by
\begin{equation}
\hat{l}^{\gamma}(\tilde{\gamma}):=\hat{l}^{\gamma}(y;\tilde{\gamma}%
)=\left(y-\frac{\xi(\tilde{\gamma})}{\gamma}\right)_+,\label{l* gamma}%
\end{equation}
and the value function is given by%
\begin{equation}
V_{C}^{\gamma}(\tilde{\gamma})=x_C-  \lambda T \left(\mathbb  E[Y]+\xi(\tilde{\gamma})\int_{\frac{\xi(\tilde{\gamma})}{\gamma}}^{\infty}  \left(y-\frac{\xi(\tilde{\gamma})}{\gamma}\right)  \mathrm{d} F(y)+\frac{\gamma}{2}   \int_0^{\infty}  \left(y\wedge  \frac{\xi(\tilde{\gamma})}{\gamma}\right) ^2  \mathrm{d} F(y)\right).\label{VP gamma2}%
\end{equation}

\end{theorem}
\begin{proof}
By  \eqref{eq:ex}  and  \eqref{XC gamma}, the customer's surplus evolves as follows:
$$
\mathrm{d}X_{C}^{\gamma}(t,l;\tilde{\gamma})=-\lambda\int_{0}^{\infty}
(1+\xi(\tilde{\gamma}))l(t,y)  \mathrm{d}F(y)\mathrm{d}t-\int_{0}^{\infty}(y-l(t,y))N(\mathrm{d}%
t,\mathrm{d}y).
$$
We adopt the following ansatz for the value functions $V$ and $g$ in Theorem \ref{veri_thm}
$$
V(t,x)=x+B(t)\text{\quad and\quad}g(t,x)=x+b(t),
$$
with $B(T)=b(T)=0$. Substituting the above two forms into the extended HJB 
equation \eqref{HJB 0} yields%
\begin{align*}
0  &  =B^{\prime}(t)-\lambda\inf_{l\in
\mathbb{A}_l
}\int_{0}^{\infty}\left(  (1+\xi(\tilde{\gamma}))l(t,y)
+\left(  y-l(t,y)\right)  +\frac{\gamma}{2}\left(  y-l(t,y)\right)
^{2}\right)  \mathrm{d} F(y)\\
&  =B^{\prime}(t)-\lambda\inf_{l\in
\mathbb{A}_l
}\int_{0}^{\infty}\left(  y+\xi(\tilde{\gamma})
l(t,y)+\frac{\gamma}{2}(y-l(t,y))^{2} \right)  \mathrm{d} F(y).
\end{align*}
It follows that the minimizer, i.e. the equilibrium strategy, is given by%
$$
\hat{l}^\gamma( \tilde{\gamma})=\left(y-\frac{\xi(\tilde{\gamma})}{\gamma}\right)_+.
$$
By the boundary condition $B(T)=0$, one obtains \eqref{VP gamma2}.
\end{proof}

We then impose the truth-telling constraint \eqref{TL gamma}, whose first-order condition is 
\begin{equation}\label{FOC_TT}
\left. \frac{\partial V_{C}^{\gamma}(\tilde{\gamma})}{\partial \tilde{\gamma}} \right|_{\tilde{\gamma}=\gamma} = 0.
\end{equation}
Differentiating \eqref{VP gamma2} with respect to $\tilde{\gamma}$ and evaluating at $\tilde{\gamma} = \gamma$ yields
\begin{equation}\label{FOC_explicit}
\int_{\frac{\xi(\gamma)}{\gamma}}^{\infty} \xi'(\gamma) \Bigl( y - \frac{2\xi(\gamma)}{\gamma} \Bigr) \, \mathrm{d}F(y)
+ \frac{\xi(\gamma)\,\xi'(\gamma)}{\gamma} \left( 1 - F\left( \frac{\xi(\gamma)}{\gamma} \right) \right) = 0.
\end{equation}
This simplifies to
$$
\xi^{\prime}(\gamma) \mathbb{E}\left[\left(Y-\frac{\xi(\gamma)}{\gamma}\right)_+\right] =0,
$$
which implies that
\begin{equation}\label{eq: fC}
\xi^{\prime}(\gamma)=0 \quad\text{for all }\tilde{\gamma}\in[\gamma_L,\gamma_H].
\end{equation} 
Hence, the risk loading factor associated with truth-telling will remain constant
for all degrees of risk aversion $\tilde{\gamma}$.  

The next theorem summarizes the optimal loading factor, optimal menu of
contracts with risk-attitude uncertainty, and verifies the truth-telling constraint.

\begin{theorem}
    
\label{thm IN gamma} \hspace{1cm}

(1) An optimal loading factor to 
Problem \ref{prob IN gamma} 
subject to \eqref{eq: fC} is given by
\begin{equation}
\hat{\xi}=\mathrm{argmax}_{\xi\geq0} \int_{\gamma_{L}}^{\gamma_{H}%
} \left(\xi \mathbb  E\left[\left(Y- \frac{\xi}{\gamma}\right)_+\right]-\frac{\gamma_I}{2}  \mathbb E\left[\left(Y- \frac{\xi}{\gamma}\right)^2_+\right] \right) \mathrm{d}G(\gamma)
.\label{eta* gamma}%
\end{equation}

(2) An\ optimal menu of contracts with loss coverage $\hat{l}(\gamma):=\hat
{l}^{\gamma}(y;\gamma)$ and premium rate $\hat{p}(\gamma):=\hat{p}(\hat
{l}(\gamma);\gamma)$ for $\gamma\in\lbrack\gamma_{L},\gamma_{H}]$ is given by
\begin{equation}
\left\{
\begin{array}
[c]{l}%
\hat{l}(\gamma)=\left(y-  \frac{\hat{\xi}}{\gamma}\right)_+    ,\\
\hat{p}(\gamma)=(1+ \hat{\xi})\lambda\mathbb  E\left[\left(Y- \frac{\hat{\xi}}{\gamma}\right)_+\right] ,
\end{array}
\right.  \label{opt menu gamma}%
\end{equation}
where $\hat{\xi}$ is given in (\ref{eta* gamma}).

(3) The insurer's optimal value function is given by
$$\begin{aligned}
&\int_{\gamma_{L}}^{\gamma_{H}} \left(\mathbb{E}\left[  X^\gamma_{I}(T,\hat{l}^{\gamma
}(\gamma);\gamma)\right]   -\frac{\gamma_I}{2} \mathrm{Var}   [X^\gamma_{I}(T,\hat{l}^{\gamma}(\gamma);\gamma)]\right)\mathrm{d}G(\gamma) \\&=\lambda T%
\int_{\gamma_{L}}^{\gamma_{H}%
}\left(\hat{\xi} \mathbb  E\left[\left(Y- \frac{\hat{\xi}}{\gamma}\right)_+\right]-\frac{\gamma_I}{2}  E\left[\left(Y- \frac{\hat{\xi}}{\gamma}\right)^2_+\right]  \right)\mathrm{d}G(\gamma
).
\end{aligned}
$$

(4) The\ optimal menu of contracts $(\hat{l}(\gamma),\hat{p}(\gamma
))_{\gamma\in\lbrack\gamma_{L},\gamma_{H}]}$ satisfies the truth-telling
constraint in the sense that
$$
\gamma=\mathrm{argmax}_{\tilde{\gamma}\in\lbrack\gamma_{L},\gamma_{H}%
]}\left\{  \mathbb{E}\left[  X^\gamma_{C}(T,\hat{l}(\tilde{\gamma});\tilde{\gamma
})\right]  -\frac{\gamma}{2}\mathrm{Var}\left[  X^\gamma_{C}(T,\hat{l}(\tilde{\gamma
});\tilde{\gamma})\right]  \right\}  ,
$$
where $\{X^\gamma_{C}(t,\hat{l}(\tilde{\gamma});\tilde{\gamma})\}_{t\in\lbrack0,T]}$
is the customer's surplus process given by   \eqref{XC gamma} if he chooses the contract $(\hat{l}%
(\tilde{\gamma}),\hat{p}(\tilde{\gamma}))$ for risk aversion $\tilde{\gamma}$.
\end{theorem}

    \begin{proof}
 (1) Combinging (\ref{eq:ex}), (\ref{XI gamma}), and (\ref{l* gamma}), the insurer's
surplus follows the dynamics%
\begin{align*}
\mathrm{d}X^\gamma_{I}(t,\hat{l}^{\gamma}(\gamma);\gamma)  &  =\lambda(1+\xi)\int
_{0}^{\infty}\left(y-\frac{\xi}{\gamma}\right)_+%
 \mathrm{d}F(y)\mathrm{d}t-\int_{0}^{\infty
} \left(y-\frac{\xi}{\gamma}\right)_+ N (\mathrm{d}t,\mathrm{d}y).
\end{align*}
It follows that%
\begin{equation}\begin{aligned}
&\int_{\gamma_{L}}^{\gamma_{H}}\left(\mathbb{E}\left[  X^\gamma_{I}(T,\hat{l}^{\gamma
}(\gamma);\gamma)\right]  -\frac{\gamma_I}{2} \mathrm{Var}   [X^\gamma_{I}(T,\hat{l}^{\gamma}(\gamma);\gamma)]\right)\mathrm{d}G(\gamma) \\&  = \lambda T\int_{\gamma_{L}}^{\gamma_{H}%
} \left( \xi \mathbb  E\left[\left(Y- \frac{\xi}{\gamma}\right)_+\right]-\frac{\gamma_I}{2}  \mathbb E\left[\left(Y- \frac{\xi}{\gamma}\right)^2_+\right] \right) \mathrm{d}G(\gamma
).
\end{aligned}\label{VI00}\end{equation}
By \eqref{eq: fC}, the risk loading factor is  constant, and thus 
\begin{align*}
\hat{\xi}  &  =\mathrm{argmax}_{\xi\geq0}\int_{\gamma_{L}}^{\gamma_{H}}\left(\mathbb{E}\left[  X^\gamma_{I}(T,\hat{l}^{\gamma
}(\gamma);\gamma)\right]  -\frac{\gamma_I}{2} \mathrm{Var}   [X^\gamma_{I}(T,\hat{l}^{\gamma}(\gamma);\gamma)]\right)\mathrm{d}G(\gamma) \\
&  =\mathrm{argmax}_{\xi\geq0}\int_{\gamma_{L}}^{\gamma_{H}%
}\xi \left(\mathbb  E\left[\left(Y- \frac{\xi}{\gamma}\right)_+\right]-\frac{\gamma_I}{2}  \mathbb E\left[\left(Y- \frac{\xi}{\gamma}\right)^2_+\right] \right) \mathrm{d}G(\gamma
).
\end{align*}

(2) This is an immediate result of \eqref{eq:ex} and  \eqref{l* gamma}.

(3) This is an immediate result of \eqref{eta* gamma} and  \eqref{VI00}.

(4) If a customer with risk aversion $\gamma$ selects the contract for type
$\tilde{\gamma}$, i.e., $(\hat{l}(\tilde{\gamma}),\hat{p}(\tilde{\gamma}))$,
the customer's surplus follows the dynamics%
\begin{align*}
\mathrm{d}X^\gamma_{C}(t,\hat{l}(\tilde{\gamma});\tilde{\gamma})  &  =-\lambda(1+\hat \xi)
\int_{0}^{\infty}\left( x-\frac{\hat{\xi}}{\tilde \gamma}  \right)_+  \mathrm{d}F(y)\mathrm{d}t-\int_{0}^{\infty}\left(y \wedge \frac{\hat\xi}{\tilde \gamma}\right)_+N (\mathrm{d}t,\mathrm{d}y).
\end{align*}
Correspondingly, the customer's value function is equal to%
\begin{align*}
J^{\gamma}(\tilde{\gamma})  &  :=\mathbb{E}\left[  X^\gamma_{C}(T,\hat{l}%
(\tilde{\gamma});\tilde{\gamma})\right]  -\frac{\gamma}{2}\mathrm{Var}\left[
X^\gamma_{C}(T,\hat{l}(\tilde{\gamma});\tilde{\gamma})\right] \\
&  =x_C-\lambda T \mathbb  E[Y]- \lambda T \int_{\frac{\hat\xi}{\tilde\gamma}}^{\infty} \hat\xi \left(y-\frac{\hat\xi }{\tilde\gamma}\right)  \mathrm{d} F(y)-\frac{\gamma}{2} \lambda T   \int_0^{\infty}  \left(y\wedge  \frac{\hat\xi}{\tilde \gamma}\right) ^2  \mathrm{d} F(y)
\end{align*}
Differentiating $J^{\gamma}(\tilde{\gamma}) $ with respect to $\tilde{\gamma}$ yields
\begin{align*}
\frac{\partial J^{\gamma}(\tilde{\gamma})}{\partial\tilde{\gamma}}  &
=-\lambda T \frac{\hat \xi^2}{\tilde\gamma^2} F\left(\frac{\hat \xi}{\tilde\gamma}\right) \left(\frac{\gamma}{\tilde\gamma} -1 \right).
\end{align*}
It is obvious to see that $\frac{\partial J^{\gamma}(\tilde{\gamma})}%
{\partial\tilde{\gamma}}>0$  for $\tilde{\gamma}<\gamma$, and $\frac{\partial J^{\gamma}(\tilde{\gamma})}%
{\partial\tilde{\gamma}}<0$  for $\tilde{\gamma}>\gamma$.  Thus
$$
\gamma=\mathrm{argmax}_{\tilde{\gamma}\in\lbrack\gamma_{L},\gamma_{H}%
]}J^{\gamma}(\tilde{\gamma}),
$$  which we completes the proof.
\end{proof}

Theorem \ref{thm IN gamma} provides a characterization of the optimal contract menu when the insurer faces uncertainty about the customer’s risk attitude. The optimal menu
$
(\hat{l}(\gamma), \hat{p}(\gamma))_{\gamma \in [\gamma_{L}, \gamma_{H}]},
$
presented in \eqref{opt menu gamma}, is composed of corresponding loss coverage rules and premium schedules, and is shown to satisfy the truth-telling requirement. An important observation is that the optimal risk loading $\hat{\xi}$ remains unchanged for all admissible levels of risk aversion. This suggests that a
{linear pricing} strategy is employed to elicit information regarding the customer's true risk aversion.
Moreover, in the special case where the insurer has perfect information on the customer’s risk aversion parameter $\gamma$, we obtain the following corollary.  In the absence of uncertainty, the problem naturally simplifies, and the contract menu reduces to a single contract.

\begin{corollary}\label{cor: no uncertainty theta}
Suppose the insurer has no uncertainty about the customer's risk aversion parameter $\gamma$.
Then, the optimal menu contains a single contract given by
$$
\left\{
\begin{array}
[c]{l}%
\hat{l}=\left(y-  \frac{\hat{\xi}}{\gamma}\right)_+ ,\\
\hat{p}=(1+ \hat{\xi})\mathbb  E\left[\left(Y- \frac{\hat{\xi}}{\gamma}\right)_+\right],
\end{array}
\right.
$$
where $$\hat{\xi}=\mathrm{argmax}_{\xi\geq0}
 \left\{\xi \mathbb  E\left[\left(Y- \frac{\xi}{\gamma}\right)_+\right]-\frac{\gamma_I}{2}  \mathbb E\left[\left(Y- \frac{\xi}{\gamma}\right)^2_+\right] \right\}.$$ 
\end{corollary}

Note that Corollary \ref{cor: no uncertainty theta} coincides with Proposition 4.2 in \cite{chen2019stochastic} where no uncertainty is considered and the problem of  Stackelberg equilibrium  is studied as presented in Section \ref{sec:2}.

The optimal contract under uncertainty in risk attitude exhibits several notable properties, summarized as follows. Their proofs directly follow from \eqref{eta* gamma} and \eqref{opt menu gamma}.

\begin{proposition}  
\label{prop opt gamma}  

(1) $\hat{\xi}$ is independent of  $\gamma$.
   
(2) $\hat{l}(\gamma)$  is increasing in $\gamma$.

(3) $\hat{p}(\gamma)$  is increasing in $\gamma$.
 \end{proposition} 

As noted above, Proposition~\ref{prop opt gamma}(1) shows that the optimal risk loading $\hat{\xi}$ is constant across all admissible levels of risk aversion. In contrast, as we will demonstrate in Section~\ref{sec:4}, when the uncertainty lies in customers’ risk types, which directly influence their expected losses, the optimal risk loading becomes a nonlinear function, and high-risk customers are  offered price discounts as part of the insurer’s screening mechanism. The linear  pattern for risk-attitude uncertainty is expected because risk attitude influences only the willingness to purchase insurance and does not directly impact expected claims or the insurer’s profit, thereby providing limited informational value and not warranting differentiated loading.
Propositions \ref{prop opt gamma}(2) and (3) formalize the natural monotonicity: more risk-averse customers optimally demand higher coverage and correspondingly pay higher premiums, in line with standard economic intuition.

\section{Risk-type uncertainty}\label{sec:4}
\subsection{Modelling}

We now consider the case in which the insurer observes the customer’s true risk-attitude parameter $\gamma$ but faces uncertainty regarding the customer’s true risk type. The risk type is represented by a parameter $\theta$, which governs the distribution of individual claim sizes and thereby reflects the customer’s underlying risk. We model $\theta$ as a random variable with cumulative distribution function $G(\theta)$ supported on a compact interval $[\theta_L, \theta_H] \subset \mathbb{R}$. 
For any function $h(x_1,\dots,x_n)$,  we use the notation $\partial_{x_i} h = {\partial h}/{\partial x_i}$ to denote the partial derivative with respect to the variable $x_i$.

Let $\{Y_i^\theta\}_{i\in\mathbb{N}}$ be a sequence of i.i.d. positive random variables representing individual claims for a customer of type $\theta$, independent of the counting process $\{N(t)\}_{t\ge 0}$. Let $F(y;\theta)$ denote the cumulative distribution function of a single claim. 
Assume that $F(\cdot;\theta)$ is differentiable with density
$
f(y;\theta)=\partial_y F(y;\theta), ~y>0 .
$
The cumulative claims up to time $t$ are given by
$$
\sum_{i=1}^{N(t)} Y_i^\theta = \int_0^t \int_0^\infty y \, N^\theta(\mathrm{d}s, \mathrm{d}y).
$$

The following assumption will be used in Theorem~\ref{thm:optimal_stop_loss}. 

\begin{assumption}\label{ass:FOD}
For all $y \ge 0$, the distribution function $F(y;\theta)$ is continuously differentiable with respect to $\theta$ and satisfies $\partial_\theta F(y;\theta) < 0$.
\end{assumption}

Let $Y_\theta$ denote a random variable with distribution function $F(\cdot;\theta)$. 
Assumption \ref{ass:FOD} implies that for any $\theta' > \theta$,
\[
F(y;\theta') \le F(y;\theta), \quad \forall\, y \ge 0.
\]
Hence, $Y_{\theta'}$ first-order stochastically dominates $Y_\theta$. 
In other words, individuals with larger $\theta$ have stochastically larger losses and are therefore riskier. 
The derivative condition in Assumption \ref{ass:FOD} is satisfied by many commonly used parametric families of loss distributions.

\begin{example}\label{special_dist}
\begin{itemize}
\item[(i)] {Uniform distribution} on $(0, \theta)$:
$$
F(y; \theta) =
\begin{cases}
0, & y \leq 0, \\
\frac{y}{\theta}, & 0 < y < \theta, \\
1, & y \geq \theta,
\end{cases}
\quad \Longrightarrow \quad
\partial_\theta F(y; \theta) = -\frac{y}{\theta^2} \leq 0.
$$

\item[(ii)] {Exponential distribution} with mean $\theta$:
$$
F(y; \theta) = 1 - e^{-y/\theta}, \quad y > 0,
\quad \Longrightarrow \quad
\partial_\theta F(y; \theta) = -\frac{y}{\theta^2} e^{-y/\theta} \leq 0.
$$

\item[(iii)] {Pareto distribution} with shape parameter $1/\theta$ and scale $y_m > 0$:
$$
F(y; \theta) = 1 - \left( \frac{y_m}{y} \right)^{1/\theta}, \quad y \geq y_m,
\quad \Longrightarrow \quad
\partial_\theta F(y; \theta) = \frac{1}{\theta^2} \left( \frac{y_m}{y} \right)^{1/\theta} \ln\left( \frac{y_m}{y} \right) \leq 0.
$$
\end{itemize}
\end{example}

Given a true risk type $\theta$, if the customer chooses a type-$\tilde{\theta}$ contract with coverage function $l$, the {customer's surplus} is
\begin{align}
X_{C}^{\theta}(t,l;\tilde{\theta}) 
&= x_{C} - p(l; \tilde{\theta}) t 
- \int_{0}^{t} \int_{0}^{\infty} ( y - l(s,y) ) N^{\theta}(\mathrm{d}s,\mathrm{d}y), \label{XP}
\end{align}
and the insurer's profit is 
\begin{align}
X_{I}^{\theta}(t,l;\tilde{\theta}) 
&= x_I+p(l; \tilde{\theta}) t 
- \int_{0}^{t} \int_{0}^{\infty} l(s,y) N^{\theta}(\mathrm{d}s,\mathrm{d}y). \label{XI}
\end{align} Under the expected-value premium principle, the premium rate is given by
\begin{equation}
p(l; \tilde{\theta}) 
= \left( 1 + \xi(\tilde{\theta}) \right) \lambda \int_{0}^{\infty} l(t,y) f(y; \tilde{\theta}) \, \mathrm{d}y.
\label{premium2}
\end{equation}
The set of admissible loading functions can be revised to 
$$
\mathcal{A}_\xi = \left\{ \xi : \xi(\tilde{\theta}) \geq 0, \quad \forall \tilde{\theta} \in [\theta_L, \theta_H] \right\}.
$$  We present the customer's and insurer's problems as follows.
\begin{problem}[Customer's problem]\label{prob ph theta}
For a type-$\theta$ customer selecting a contract designed for type $\tilde{\theta}$, the value function is
\begin{equation}
V_{C}^{\theta}(\tilde{\theta}) 
:= \sup_{l \in \mathbb{A}_{l}} 
\left\{ \mathbb{E}\left[ X_{C}^{\theta}(T,l;\tilde{\theta}) \right] 
- \frac{\gamma(\theta)}{2} \mathrm{Var} \left[ X_{C}^{\theta}(T,l;\tilde{\theta}) \right] \right\}.
\label{VP 12}
\end{equation}
The corresponding optimal coverage is denoted by 
$\hat{l}^{\theta}(\tilde{\theta}) = \{\hat{l}^{\theta}(y;\tilde{\theta})\}_{y>0}$.
\end{problem}

Here we allow the risk-aversion parameter to depend on the risk type, i.e., $\gamma=\gamma(\theta)$. 
If $\gamma(\theta)$ is increasing in $\theta$, then customers with higher risk types are more risk-averse; 
if $\gamma(\theta)$ is decreasing in $\theta$, higher-risk customers are less risk-averse.

The insurer aims to maximize her weighted MV utility across all priors of
the customer's risk type, while also seeking to elicit the true
information regarding the customer's risk type. Condition (\ref{TL theta})
is the corresponding truth-telling constraint for risk-type uncertainty.
\begin{problem}[Insurer's problem]\label{prob in theta}
The insurer's objective is to maximize its weighted MV utility:
$$
\sup_{\xi(\cdot) \in \mathbb{A}_{\xi}} 
\int_{\theta_{L}}^{\theta_{H}} 
\left[ \mathbb{E} \left[ X_{I}^{\theta}(T,\hat{l}^{\theta}(\theta); \theta) \right] 
- \frac{\gamma_I}{2} \mathrm{Var} \left[ X_{I}(T,\hat{l}^{\gamma}(\gamma); \gamma) \right] \right] 
\, \mathrm{d}G(\theta),
$$
subject to the \textit{truth-telling constraint}:
\begin{equation}
\theta \in \operatorname{argmax}_{\tilde{\theta} \in [\theta_{L},\theta_{H}]} 
V_{C}^{\theta}(\tilde{\theta}).
\label{TL theta}
\end{equation}
\end{problem}

\subsection{Verification theorem}

The verification theorem  is similar to that in Section \ref{sec:vt}.  We first give the   definition of equilibrium strategy below. For $(t,x)\in
\lbrack0,T]\times%
\mathbb{R}
$, we define%
\begin{equation*}
J^{\theta}(t,x,l;\tilde{\theta})=\mathbb{E}\left[  \left.  X_{C}^{\theta
}(T,l;\tilde{\theta})\right\vert X_{C}^{\theta}(t,l;\tilde{\theta})=x\right]
-\frac{\gamma(\theta)}{2}\mathrm{Var}\left[  \left.  X_{C}^{\theta}(T,l;\tilde{\theta
})\right\vert X_{C}^{\theta}(t,l;\tilde{\theta})=x\right].\label{J}
\end{equation*}

\begin{definition}
\label{def equilibrium1} Given $(t,x)\in\lbrack0,T]\times%
\mathbb{R}
$ and $\hat{l}\in\mathbb{A}_{l}$, we define a perturbed strategy
$l^{\varepsilon}$ as
\begin{equation*}
l^{\varepsilon}(s,y)=\left\{
\begin{array}
[c]{ll}%
\bar{l}(y), & s\in\lbrack t,t+\varepsilon),\text{ }y\in%
\mathbb{R}
\text{,}\\
\hat{l}(s,y), & s\in\lbrack t+\varepsilon,T],\text{ }y\in%
\mathbb{R}
\text{,}%
\end{array}
\right.  \label{l eps1}%
\end{equation*}
where $\bar{l}:%
\mathbb{R}
\rightarrow%
\mathbb{R}
$ and $\varepsilon>0$. Suppose that, for any given $(t,x)\in\lbrack0,T]\times%
\mathbb{R}
$, and any such perturbed strategy $l^{\varepsilon}\in\mathbb{A}_{l}$, we
have
$$
\liminf_{\varepsilon\rightarrow0}{\frac{J^{\theta}(t,x,\hat{l};\tilde{\theta
})-J^{\theta}(t,x,l^{\varepsilon};\tilde{\theta})}{\varepsilon}}\geq0.
$$
Then, $\hat{l}$ is called an equilibrium strategy for the dynamic MV problem
(\ref{VP 12}).
\end{definition}
For any $l\in\mathbb{A}_{l}$ and {$\varphi(t,x)\in C^{1,2}([0,T]\times
\mathbb{R})$, i.e., first-order continuously differentiable in $t$ and
second-order continuously differentiable in }$x${, we define the
integro-differential operator}%
\begin{equation}\begin{aligned}
\mathcal{L}_{l}\varphi(t,x) : &  ={\frac{\partial\varphi(t,x)}{\partial
t}-\frac{\partial\varphi(t,x)}{\partial x}}\lambda\int_{0}^{\infty}
(1+\xi(\tilde{\theta}))l(t,y))  f(y;\tilde{\theta
})\mathrm{d}y\\
  &+\lambda\int_{0}^{\infty}\left[  \varphi(t,x-y+l(t,y))-\varphi(t,x)\right]
f(y;\theta)\mathrm{d}y.\label{operator1}
\end{aligned}
\end{equation}

\begin{theorem} Let the integro-differential operator $\mathcal{L}_l$ be defined as in \eqref{operator1}.  Suppose there exist $V(t,x), g(t,x)\in C^{1,2}([0,T])\times\mathbb R$ satisfy the following  conditions:
\begin{itemize} 
\item[(1)] For any $(t,x)\in[0,T]\times\mathbb R$, \begin{equation}
\sup_{l\in \mathbb{R}
}\left\{  \mathcal{L}_{l}V(t,x)-\frac{1}{2}\gamma(\theta)\mathcal{L}_{l}%
g^{2}(t,x)+\gamma(\theta) g(t,x)\mathcal{L}_{l}g(t,x)\right\}  =0 \label{HJB 1},%
\end{equation}
and $\hat{l}$ denotes the optimal value to achieve the supremum in $l$. 
\item[(2)] For any $(t, x) \in[0, T] \times \mathbb{R}$,
\begin{equation*}
\left\{\begin{array}{l}
\mathcal{L}_{\hat{l}}g(t,x)=0,  \\
V(T,x)=g(T,x)=x.
\end{array}\right.\label{equiva1}
\end{equation*}
\item[(3)] For any $(t, x) \in[0, T] \times \mathbb{R}$, $\hat l(t)$, $\mathcal{L}_{\hat l}V(t,x)$ and $\mathcal{L}_{\hat l}g^2(t, x)$ are  deterministic functions of $t$ and independent of $x$.
\end{itemize}
Then $\hat{l}$  is the equilibrium strategy  and $V(t,x)=J^{\theta}(t,x,\hat{l};\tilde{\theta})$ is the equilibrium value function to the dynamic MV problem \eqref{VP 12}. Besides, $g(t,x)=\mathbb{E}%
\left[  \left.  X_{C}^{\theta}(T,\hat{l};\tilde{\theta})\right\vert
X_{C}^{\theta}(t,\hat{l};\tilde{\theta})=x\right]$. \label{veri_thm1}
\end{theorem}

\subsection{Solution of the optimal menu}
 We determine the optimal menu of insurance contracts by following the four-step procedure introduced in Section \ref{sec:3.3}.   
 Firstly, following a similar approach to that in Theorem \ref{thm PH gamma}, the optimal loss coverage function for the customer's problem \eqref{VP 12} is given by
 the following proposition. 
 \begin{proposition}\label{prop:4.1}
     An optimal insurance coverage to the customer's problem
(\ref{VP 12}) is given by
\begin{equation}\label{eq:l_0}
\hat{l}^{\theta}(y; \tilde{\theta}) = \left( y + \frac{1}{\gamma(\theta)} - \frac{f(y;\tilde \theta)(1 + \xi(\tilde{\theta}))}{\gamma(\theta) f(y; \theta)} \right) \vee 0 \wedge y,
\end{equation}
and the value function is given by%
\begin{equation}
\begin{aligned}
V_{C}^{\theta}(\tilde{\theta}) &  =x_{C}-\lambda T\left[\int_{0}^{\infty}\left[
 (1+\xi(\tilde{\theta}))  \hat
{l}^{\theta}(y;\tilde{\theta})\right]  f(y;\tilde{\theta})\mathrm{d}y\right.\\
&\left.  +\int_{0}^{\infty}\left(  \left(  y-\hat{l}^{\theta}%
(y;\tilde{\theta})\right)  +\frac{\gamma(\theta)}{2}\left(  y-\hat{l}^{\theta
}(y;\tilde{\theta})\right)  ^{2}\right)  f(y;\theta)\mathrm{d}y\right].
\end{aligned}\label{eq:general}
\end{equation}
\end{proposition} 
\begin{proof} By  (\ref{XP}) and  (\ref{premium2}), the customer's surplus evolves as follows:
$$
\mathrm{d}X_{C}^{\theta}(t,l;\tilde{\theta})=-\lambda\int_{0}^{\infty}
(1+\xi(\tilde{\theta}))l(t,y)  f(y;\tilde{\theta
})\mathrm{d}y\mathrm{d}t-\int_{0}^{\infty}(y-l(t,y))N^{\theta}(\mathrm{d}%
t,\mathrm{d}y).
$$
We aim to solve for $V$ and $g$ satisfying all the conditions of Theorem \ref{veri_thm1}. 
Again, we use the following ansatz for the value functions:
\begin{equation}
V(t,x)=x+B(t)\text{\quad and\quad}g(t,x)=x+b(t), \label{ans theta1}%
\end{equation}
with $B(T)=b(T)=0$. Substituting (\ref{ans theta1}) into the extended HJB
equation (\ref{HJB 1}) yields%
\begin{equation*}
B^{\prime}(t)-\lambda\inf_{l\in
\mathbb{R}
}\int_{0}^{\infty}\left\{ (1+\xi(\tilde{\theta}))l(t,y) f(y;\tilde{\theta})+\left[  \left(  y-l(t,y)\right)
+\frac{\gamma(\theta)}{2}\left(  y-l(t,y)\right)  ^{2}\right]  f(y;\theta)\right\}
\mathrm{d}y=0. \label{B}%
\end{equation*}
It follows that the minimizer, that is the equilibrium strategy, is given by%
$$
\hat{l}^{\theta}(y; \tilde{\theta}) = \left( y + \frac{1}{\gamma(\theta)} - \frac{f(y;\tilde \theta)(1 + \xi(\tilde{\theta}))}{\gamma(\theta) f(y; \theta)} \right) \vee 0 \wedge y.
$$
By the boundary condition $B(T)=0$, one obtains \eqref{eq:general}.
\end{proof}

The customer's optimal strategy in \eqref{eq:l_0} is explicit, but its complexity prevents a direct analytical verification of the truth-telling constraint \eqref{TL theta} via the first-order condition. According to \eqref{eq:l_0}, when the customer reports truthfully, i.e., $\tilde{\theta}=\theta$, we obtain
\begin{equation}\label{ltt0}
\hat{l}^{\theta}(y;\theta) = \left(y - \frac{\xi(\theta)}{\gamma(\theta)}\right)_+,
\end{equation}
which corresponds to an excess-of-loss insurance strategy. To facilitate a tractable analysis, we therefore restrict attention to strategies of the excess-of-loss form.  Moreover, in the classical MV framework without information asymmetry, it is well known that an excess-of-loss strategy is optimal under the expected-value premium. Hence, this assumption is also reasonable in our context.

Accordingly, given a true risk type $\theta$, we assume that if the customer selects a type-$\tilde{\theta}$ contract with coverage function $l$  taking the form of an excess-of-loss contract:
\begin{equation}\label{eq:stop_loss}
l(t,y;\tilde{\theta}) = (y - d(t,y;\tilde{\theta}))_+, \quad d(t,y;\tilde{\theta}) \ge 0,
\end{equation}
where $d(t,y;\tilde{\theta})$ denotes the customer’s retention level, which may depend on time $t$, the realized loss $y$, and the type parameter $\tilde{\theta}$.

\begin{theorem}\label{thm:optimal_stop_loss}
When the coverage takes the excess-of-loss form in \eqref{eq:stop_loss}, 
an optimal solution to the customer’s problem \eqref{VP 12} is given by
$$
\hat l^\theta(y;\tilde{\theta})=(y-\hat d^\theta(\tilde{\theta}))_+,
$$
where
\begin{equation*}\label{eq:d}
\hat{d}^\theta(\tilde\theta) = \left(\frac{1+\xi(\tilde{\theta})}{\gamma} \cdot \frac{1 - F(\hat{d}^\theta(\tilde\theta); \theta)}{1 - F(\hat{d}^\theta(\tilde\theta); \tilde{\theta})} - \frac{1}{\gamma}\right)_+.
\end{equation*}
The corresponding value function is
\begin{equation}\label{eq:value_function}
\begin{aligned}
V_C^{\theta}(\tilde{\theta}) &= x_C 
 - \lambda T \left[(1+\xi(\tilde{\theta}))\int_{\hat{d}^\theta(\tilde\theta)}^\infty  (1- F(y; \tilde{\theta})) \, \mathrm d y  + \int_0^{\hat{d}^\theta(\tilde\theta)} \left( y + \frac{\gamma(\theta)}{2} y^2 \right) f(y; \theta) \, \mathrm d y \right.\\
& \left.+ \left( \hat{d}^\theta(\tilde\theta) + \frac{\gamma(\theta)}{2} (\hat{d}^\theta(\tilde\theta))^2 \right) \left(1 - F(\hat{d}^\theta(\tilde\theta); \theta) \right)\right].
\end{aligned}
\end{equation}
In particular, when $\tilde{\theta} = \theta$,
\begin{equation}\label{eq:d_hat}
\hat l^\theta(y;{\theta})=\left(y-\frac{\xi(\theta)}{\gamma(\theta)}\right)_+,
\end{equation}
and
\begin{equation}\label{VP true}
\begin{aligned}
V_C^{\theta}(\theta) &= x_C - \lambda T \Bigg[(1+\xi(\theta))\int_{\frac{\xi(\theta)}{\gamma(\theta)}}^{\infty}(1-F(y;\theta))\mathrm d y  + \int_0^{\frac{\xi(\theta)}{\gamma(\theta)}}\left(y+\frac{\gamma(\theta)}{2}y^2\right)f(y;\theta)\mathrm d y\\
& \quad + \left(\frac{\xi(\theta)}{\gamma}+\frac{\gamma(\theta)}{2}\left(\frac{\xi(\theta)}{\gamma(\theta)}\right)^2\right)\left(1-F\left(\frac{\xi(\theta)}{\gamma(\theta)};\theta\right)\right)\Bigg].
\end{aligned}
\end{equation}
\end{theorem}

\begin{proof}
We aim to solve for $V$ and $g$ satisfying all the conditions of Theorem \ref{veri_thm1}. 
We use the same  ansatz for the value functions in \eqref{ans theta1}. Substituting (\ref{ans theta1})  and \eqref{eq:stop_loss} back into the extended HJB
equation (\ref{HJB 1}) yields%
\begin{align*}
0  &  =B^{\prime}(t)-\lambda \inf_{d\geq0}\left((1+\xi(\tilde{\theta}))
\int_{d}^{\infty} (y-d) f(y; \tilde{\theta} )\mathrm{d}y\right.\\&
\left.+\int_{0}^{d}( y   +\frac{\gamma}{2}  y^2 ) f(y; \theta )\mathrm{d}y + (d+\frac{\gamma(\theta)}{2}d^2 ) (1-F(d,\theta))\right). 
\end{align*}
It follows that the minimizer   is given implicitly by
$$ \hat{d}^\theta(\tilde\theta) = \left(\frac{1+\xi(\tilde{\theta})}{\gamma(\theta)} \cdot \frac{1 - F(\hat{d}^\theta(\tilde\theta); \theta)}{1 - F(\hat{d}^\theta(\tilde\theta); \tilde{\theta})} - \frac{1}{\gamma(\theta)}\right)_+.$$
By the boundary condition $B(T)=0$, one obtains \eqref{eq:value_function}. 
In particular, when $\tilde{\theta}=\theta$, we have
$$\hat{d}^{\theta}(y;\theta)=\frac{\xi({\theta})}{\gamma(\theta)}\geq 0,$$
and $V_C^{\theta}({\theta})$ 
 is given by \eqref{VP true}. 
\end{proof}

The following lemma characterizes the risk-loading function $\xi(\theta)$ that satisfies the truth-telling constraint. In particular, $\xi(\theta)$ must solve an ODE, which can also be written in an implicit integral form.

Define
\begin{equation}\label{eq:Psi}
\Psi (\theta ,\xi ):=\int_{\xi /\gamma(\theta) }^{\infty }(1-F(y;\theta ))\,\mathrm{d}y,
\qquad 
\Psi_{\theta}(\theta,\xi)
:
=
-\int_{\xi/\gamma(\theta)}^{\infty}
\partial_{\theta}F(y;\theta)\,\d y .
\end{equation}
\begin{lemma}\label{prop:xi-implicit}
The truth-telling first-order condition implies that $\hat{\xi}(\theta)$ solves the ODE
\begin{equation}
\hat{\xi}^{\prime }(\theta )\Psi (\theta ,\hat{\xi}(\theta ))+(1+\hat{\xi}%
(\theta ))\Psi _{\theta }(\theta ,\hat{\xi}(\theta ))=0. \label{ODE}
\end{equation}%
\end{lemma}

\begin{proof}
We first compute the derivative:
\begin{eqnarray*}
\frac{\partial V_{C}^{\theta }(\tilde{\theta})}{\partial \tilde{\theta}}
&\propto &-\xi ^{\prime }(\tilde{\theta})\Psi (\tilde\theta ,\xi(\tilde\theta) )-(1+\xi (\tilde{\theta}%
))\Psi _{\theta }(\tilde\theta ,\xi(\tilde\theta) ) -(\gamma(\theta) \hat{d}^{\theta }(\tilde{\theta})-\xi (\tilde{\theta}))(1-F(\hat{d%
}^{\theta }(\tilde{\theta});\tilde{\theta}))\,\frac{\hat{d}^{\theta }(\tilde{%
\theta})}{\partial \tilde{\theta}}.
\end{eqnarray*}%
By letting $\tilde{\theta}=\theta $, since $\hat{d}^{\theta }(\theta )=\frac{%
\xi (\theta )}{\gamma(\theta)}$, the first-order-condition $\left. \frac{\partial
V_{C}^{\theta }(\tilde{\theta})}{\partial \tilde{\theta}}\right\vert _{%
\tilde{\theta}=\theta }=0$ yields an ODE for the
optimal risk-loading function $\hat{\xi}(\theta )$, namely \eqref{ODE}.
\end{proof}
Define
\[
\phi(\theta,\xi) := \frac{\Psi_{\theta}(\theta,\xi)}{\Psi(\theta,\xi)},
\]  
with the convention $0=0/0$.

\begin{assumption}\label{ass:phi}
The function $\phi(\theta,\xi)$ and its derivative with respect to $\xi$ are bounded on 
$[\theta_L,\theta_H]\times\mathcal A_\xi$. That is, there exist constants $K,M>0$ such that
\[
0\leq \phi(\theta,\xi)\le K,
\qquad
\left|\frac{\partial \phi}{\partial \xi}(\theta,\xi)\right|\le M,
\quad
\forall (\theta,\xi)\in[\theta_L,\theta_H]\times\mathcal A_\xi .
\]
\end{assumption}
Assumption~\ref{ass:phi} is is not particularly restrictive. 
For instance, suppose that the random variable $Y$ has bounded support and that 
$\partial_\theta F(y;\theta)$ is continuous in $y$ on this support. 
Then  $1-F(y;\theta)$ is bounded and the domain of 
integration in \eqref{eq:Psi} is finite. Consequently, both 
$\Psi(\theta,\xi)$ and $\Psi_\theta(\theta,\xi)$ are bounded on 
$[\theta_L,\theta_H]\times\mathcal A_\xi$. 
It follows that 
$\phi(\theta,\xi)=\Psi_\theta(\theta,\xi)/\Psi(\theta,\xi)$ 
is also bounded. Moreover, the boundedness of 
$\partial_\xi\phi(\theta,\xi)$ follows from the same bounded-support 
argument. Hence Assumption~\ref{ass:phi} is satisfied by many standard 
bounded distributions, such as the uniform and truncated normal distributions.

\begin{lemma}\label{thm:existence-uniqueness-final}
Under Assumptions \ref{ass:FOD} and \ref{ass:phi}, let the boundary condition $\xi(\theta_L) = \xi_0 > 0$.
Define the closed subspace
\[
\mathcal{A}_{\xi_0} = \{\xi \in \mathcal{A}_\xi : \xi(\theta_L) = \xi_0, \ 0 \le \xi(\theta) \le \xi_0\}.
\] If $(1+\xi_0)M(\theta_H - \theta_L) < 1$ and $(1+\xi_0) e^{-K(\theta_H-\theta_L)} - 1 \ge 0$, then 
 the operator $T: \mathcal{A}_{\xi_0} \to \mathcal{A}_{\xi_0}$ defined by
\[
(T\xi)(\theta) = (1+\xi_0)\exp\Big(-\int_{\theta_L}^{\theta}\phi(z,\xi(z))\,\ z\Big) - 1
\]
is a contraction mapping on $\mathcal{A}_{\xi_0}$ and thus has a unique fixed point in $\mathcal{A}_{\xi_0}$.
\end{lemma}
\begin{proof}
By Assumption \ref{ass:FOD}, $\phi(\theta,\xi) > 0$ for all $\theta \in [\theta_L,\theta_H]$. Then
\[
\exp\Big(-\int_{\theta_L}^{\theta}\phi(z,\xi(z)) \d z\Big) \le 1 \quad \implies \quad (T\xi)(\theta) \le \xi_0.
\]
For the lower bound, using Assumption \ref{ass:phi} and $(1+\xi_0) e^{-K(\theta_H-\theta_L)} - 1 \ge 0$, we have 
\[
(T\xi)(\theta) = (1+\xi_0)\exp\Big(-\int_{\theta_L}^{\theta}\phi(z,\xi(z)) \d z\Big) - 1 \ge (1+\xi_0) e^{-K(\theta_H-\theta_L)} - 1 \ge 0.\]

For any $\xi_1,\xi_2 \in \mathcal{A}_{\xi_0}$, by the mean value theorem,
for each $z$ there exists $\tilde{\xi}(z)$ between $\xi_1(z)$ and $\xi_2(z)$ such that
\[
\phi(z,\xi_1(z))-\phi(z,\xi_2(z))
=
\frac{\partial \phi}{\partial \xi}(z,\tilde{\xi}(z))
(\xi_1(z)-\xi_2(z)).
\]
Then \[
\begin{aligned}
|(T\xi_1)(\theta)-(T\xi_2)(\theta)|
&\le (1+\xi_0)\int_{\theta_L}^{\theta}
|\phi(z,\xi_1(z))-\phi(z,\xi_2(z))|\d z\\
&\le (1+\xi_0)M\int_{\theta_L}^{\theta}
|\xi_1(z)-\xi_2(z)|\mathrm dz\\
&\le (1+\xi_0)M(\theta_H-\theta_L)\|\xi_1-\xi_2\|_\infty.
\end{aligned}
\]
Taking the supremum over $\theta \in [\theta_L,\theta_H]$, we get
\[
\|T\xi_1 - T\xi_2\|_\infty \le (1+\xi_0) L (\theta_H - \theta_L) \|\xi_1 - \xi_2\|_\infty.
\]
By the assumption $(1+\xi_0)M(\theta_H - \theta_L) < 1$, the operator $T$ is a contraction.
Since the subspace $\mathcal{A}_{\xi_0}$ is closed and complete under the sup norm, the Banach fixed-point theorem guarantees a unique fixed point of $T$ in $\mathcal{A}_{\xi_0}$. 
This fixed point is therefore the unique solution to
\[
\xi(\theta) = (1+\xi_0)\exp\left(-\int_{\theta_L}^{\theta} \phi(z,\xi(z))\, dz\right) - 1.
\]
\end{proof}

The conditions in Assumption~\ref{ass:phi} are {sufficient} to guarantee the existence and uniqueness of a positive solution to \eqref{ODE}, but they are not necessary. In particular, for some specific distributions, these conditions may be violated, yet a unique positive solution to \eqref{ODE} can still be obtained. For example, in the case where $Y$ follows an exponential distribution as discussed in Section~\ref{sec:5.2},  Assumption~\ref{ass:phi} on the boundedness of $\phi$ is not satisfied, but the ODE still admits a closed-form solution.
Moreover, the two conditions
\[
(1+\xi_0)M(\theta_H - \theta_L) < 1 \quad \text{and} \quad (1+\xi_0) e^{-K(\theta_H-\theta_L)} - 1 \ge 0
\]
imply that we must require
\[
\theta_H - \theta_L \le \min \Bigg( \frac{1}{M(1+\xi_0)}, \frac{\ln(1+\xi_0)}{K} \Bigg),
\]
which means that the uncertainty set $[\theta_L, \theta_H]$ cannot be too large; otherwise, a solution $\xi$ that fully separates all agent types may not exist.
Finally, even if the condition $(1+\xi_0)M(\theta_H - \theta_L) < 1$ is not satisfied, a solution to \eqref{ODE} still exists in $\mathcal{A}_\xi$ by Schauder's fixed-point theorem, although uniqueness is no longer guaranteed.

The following theorem summarizes the optimal loading factor, optimal menu of
contracts with risk-type uncertainty, the insurer's MV utility from offering the optimal menu, and verifies the truth-telling constraint.
\begin{theorem}\label{thm:main}
Suppose that the coverage takes the excess-of-loss form given in \eqref{eq:stop_loss}.
Under Assumptions~\ref{ass:FOD}–\ref{ass:phi} and the conditions of Lemma~\ref{thm:existence-uniqueness-final}, the following results hold.

 (1) The optimal loading factor $\hat \xi(\theta,C^*)$ to  Problem \ref{prob in theta}
 subject to \eqref{ODE} is characterized as the solution to 
\begin{equation}\label{eq:xi_integral}
\hat \xi(\theta, C)
=
C
\exp\!\left(
-\int_{\theta_L}^{\theta}
\phi\!\left(z,\hat\xi(z,C)\right)\,\mathrm d z
\right)
-1.
\end{equation}
 The constant $C^*$ is chosen to maximize
$
H(C)
$
over 
$
\mathcal C
:=
\left\{
C>0:
\hat\xi(\theta,C)\ge0
\ \text{for all } \theta\in[\theta_L,\theta_H]
\right\},
$
i.e.,
$
C^*:=\arg\max_{C\in\mathcal C} H(C),
$
where
\begin{equation}\label{H}
H(C)=
\int_{\theta_{L}}^{\theta_{H}}
\int_{0}^{\infty}
\left(
\hat \xi(\theta,C)
\left(y-\frac{\hat \xi(\theta,C)}{\gamma(\theta)}\right)_+
-\frac{\gamma_I}{2}
\left(y-\frac{\hat\xi(\theta,C)}{\gamma(\theta)}\right)_+^2
\right)
f(y;\theta)\,
\mathrm d y\,\mathrm d G(\theta).
\end{equation}

(2) An\ optimal menu of contracts with loss coverage $\hat{l}(\theta):=\hat
{l}(y;\theta)$ and premium rate $\hat{p}(\theta):=\hat{p}(\hat{l}%
(\theta);\theta)$ for $\theta\in\lbrack\theta_{L},\theta_{H}]$ is given by
\begin{equation}
\left\{
\begin{array}
[c]{l}%
\hat{l}(\theta)=\left(y-\frac{\hat \xi(\theta)}{\gamma(\theta)}\right)_+,\\
\hat{p}(\theta)=(1+\hat \xi(\theta) )\lambda\mathbb  E\left[ \left(X-\frac{\hat \xi(\theta)}{\gamma(\theta)}\right)_+\right].
\end{array}
\right.  \label{opt menu theta}%
\end{equation}

(3) Correspondingly, the insurer's MV utility is given by
\begin{equation}
\int_{\theta_{L}}^{\theta_{H}}\mathbb{E}\left[  X_{I}^{\theta}(T,\hat
{l}(\theta);\theta)\right]  -\frac{\gamma_I}{2} \mathrm{Var}   [X_{I}(T,\hat{l}^{\theta}(\theta);\theta)]\mathrm{d}G(\theta)
= \frac
{\lambda T}{2} H(C^*). \label{VI theta}%
\end{equation}

(4) The\ optimal menu of contracts $(\hat
{l}(\theta),\hat{p}(\theta))_{\theta\in\lbrack\theta_{L},\theta_{H}]}$
satisfies the truth-telling constraint in the sense that
$$
\theta=\mathrm{argmax}_{\tilde{\theta}\in\lbrack\theta_{L},\theta_{H}%
]}\left\{  \mathbb{E}\left[  X_{C}^{\theta}(T,\hat{l}(\tilde{\theta}%
);\tilde{\theta})\right]  -\frac{\gamma(\theta)}{2}\mathrm{Var}\left[  X_{C}^{\theta
}(T,\hat{l}(\tilde{\theta});\tilde{\theta})\right]  \right\}  ,
$$
where $\{X_{C}^{\theta}(t,\hat{l}(\tilde{\theta});\tilde{\theta}%
)\}_{t\in\lbrack0,T]}$ is the surplus process defined in \eqref{XP} for a type-$\theta$ customer who
chooses the type-$\tilde{\theta}$ contract $(\hat{l}(\tilde{\theta}),\hat
{p}(\tilde{\theta}))$.
\end{theorem}

\begin{proof}
(1)
By \eqref{XI} and \eqref{premium2}, the insurer's surplus follows
the dynamics
\begin{align*}
\mathrm{d}X_{I}^{\theta}(t,\hat{l}^{\theta}(\theta);\theta)  &  =\lambda
\int_{0}^{\infty}\left( (1+\xi(\theta))\left(y-\frac{\xi(\theta)}{\gamma(\theta)}\right)_+\right)  f(y;\theta
)\mathrm{d}y\mathrm{d}t-\int_{0}^{\infty} \left(y-\frac{\xi(\theta)}{\gamma(\theta)}\right)_+ N^{\theta}(\mathrm{d}t,\mathrm{d}y).
\end{align*}
 It follows that%
\begin{align*}
&  \int_{\theta_{L}}^{\theta_{H}}\mathbb{E}\left[  X_{I}^{\theta}(T,\hat
{l}^{\theta}(\theta);\theta)\right]   -\frac{\gamma_I}{2} \mathrm{Var}   [X_{I}(T,\hat{l}^{\theta}(\theta);\theta)]\mathrm{d}G(\theta)\nonumber\\
&  =\int_{\theta_{L}}^{\theta_{H}}\int_{0}^{\infty}  \left(\xi(\theta)\left(y-\frac{\xi(\theta)}{\gamma(\theta)}\right)_+  -\frac{\gamma_I}{2} \left(y-\frac{\xi(\theta)}{\gamma(\theta)}\right)_+^2 \right)f(y;\theta)\mathrm{d}y \mathrm{d}G(\theta)\label{midVI}.
\end{align*} By Lemma \ref{thm:existence-uniqueness-final}, we have \begin{equation*}
\hat\xi(\theta, C^*) = C^* \exp\left( \int_{\theta_L}^\theta - \phi(z, \hat\xi(z, C^*)) \, \mathrm d z \right) - 1,
\end{equation*}  where $C^{\ast}:=\operatorname{argmax}_{C\geq 0}H(C)$ and $H$ is defined by \eqref{H}.
 
(2) This is an immediate result of  (\ref{premium2}) and (\ref{eq:d_hat}).

(3) It follows directly from  \begin{equation*}\label{eq:V0}
\sup_{\xi\in\mathbb{A}_{\xi}}\int_{\theta_{L}}^{\theta_{H}}\mathbb{E}\left[
X_{I}^{\theta}(T,\hat{l}^{\theta}(\theta);\theta)\right]  -\frac{\gamma_I}{2} \mathrm{Var}   [X_{I}(T,\hat{l}^{\theta}(\theta);\theta)]\mathrm{d}G(\theta)
= \frac
{\lambda T}{2} H(C^*).\end{equation*} 

(4)  If a type-$\theta$ customer selects the contract for type $\tilde{\theta}%
$, i.e., $(\hat{l}(\tilde{\theta}),\hat{p}(\tilde{\theta}))$, the customer's
surplus evolves as follows:
$$
\mathrm{d}X_{C}^{\theta}(t,\hat{l}(\tilde{\theta});\tilde{\theta})%
=-\lambda\int_{0}^{\infty}
(1+\hat\xi(\tilde{\theta}))\hat l(t,y;\tilde \theta)  f(y;\tilde{\theta
})\mathrm{d}y\mathrm{d}t-\int_{0}^{\infty}(y-\hat l(t,y;\tilde \theta))N^{\theta}(\mathrm{d}%
t,\mathrm{d}y).
$$
Correspondingly,  we have
\begin{align*}
J^{\theta}(\tilde{\theta})  &  =x_{0}-\lambda T (1+\hat \xi(\tilde \theta)) \int_{0}^{\infty} 
\hat{l}(t,y;\tilde{\theta}) f(y;\tilde{\theta})\mathrm{d}y\\
&  -\lambda T\int_{0}^{\infty}\left[  \left(  y-\hat{l}(t,y;\tilde{\theta
})\right)  +\frac{\gamma(\theta)}{2}\left(  y-\hat{l}(t,y;\tilde{\theta})\right)
^{2}\right]  f(y;\theta)\mathrm{d}y.
\end{align*}
Substituting $\hat{l}(\tilde{\theta})=(y- \frac{\hat \xi(\tilde \theta)}{\gamma(\theta)})_+$ into the above equation yields
$$
\begin{aligned}
J^{\theta}(\tilde{\theta}) &= x_{0} - \lambda T (1 + \hat{\xi}(\tilde{\theta})) \int_{\frac{\hat{\xi}(\tilde{\theta})}{\gamma(\theta)}}^{\infty} \left( y - \frac{\hat{\xi}(\tilde{\theta})}{\gamma(\theta)} \right) f(y; \tilde{\theta}) \, \mathrm{d}y \\
&\quad - \lambda T \left[ \int_{0}^{\frac{\hat{\xi}(\tilde{\theta})}{\gamma(\theta)}} \left( y + \frac{\gamma(\theta)}{2} y^2 \right) f(y; \theta) \, \mathrm{d}y + \frac{\hat{\xi}(\tilde{\theta})}{\gamma(\theta)} \left( 1 + \frac{\hat{\xi}(\tilde{\theta})}{2 } \right) \left( 1 - F\left( \frac{\hat{\xi}(\tilde{\theta})}{\gamma(\theta)}; \theta \right) \right) \right].
\end{aligned}
$$
Differentiating $J^{\theta}(\tilde{\theta})$ with respect to $\tilde{\theta}$ yields
$$
\begin{aligned}
\frac{\mathrm dJ^{\theta}(\tilde{\theta})}{\mathrm d\tilde{\theta}} &= - \lambda T \Bigg[ \hat{\xi}'(\tilde{\theta}) \int_{\frac{\hat{\xi}(\tilde{\theta})}{\gamma(\theta)}}^{\infty} \left( y - \frac{\hat{\xi}(\tilde{\theta})}{\gamma(\theta)} \right) f(y; \tilde{\theta}) \, \mathrm{d}y \\
&\quad - (1 + \hat{\xi}(\tilde{\theta})) \left( \frac{\hat{\xi}'(\tilde{\theta})}{\gamma(\theta)} \left( 1 - F\left( \frac{\hat{\xi}(\tilde{\theta})}{\gamma(\theta)}; \tilde{\theta} \right) \right) - \int_{\frac{\hat{\xi}(\tilde{\theta})}{\gamma(\theta)}}^{\infty} \left( y - \frac{\hat{\xi}(\tilde{\theta})}{\gamma(\theta)} \right) \frac{\partial f(y; \tilde{\theta})}{\partial \tilde{\theta}} \, \mathrm{d}y \right) \\
&\quad + \frac{\hat{\xi}'(\tilde{\theta})}{\gamma(\theta)} \left( 1 + {\hat{\xi}(\tilde{\theta})} \right) \left( 1 - F\left( \frac{\hat{\xi}(\tilde{\theta})}{\gamma(\theta)}; \theta \right) \right) \Bigg].
\end{aligned}
$$
Since we have
$$  \int_{\frac{\xi(\theta)}{\gamma(\theta)}}^{\infty} \left( \xi'({\theta})  y  f(y ; {\theta})  - \xi({\theta}) \frac{\xi(\theta)'}{\gamma(\theta)} f(y ; {\theta}) + (1+\xi({\theta})) \left(y - \frac{\xi(\theta)}{\gamma(\theta)} \right)\  \frac{\partial f(y; {\theta})}{\partial {\theta}} \right) \mathrm d y  =0,$$
it  then gives
$$ \begin{aligned} \frac{dJ^{\theta}(\tilde{\theta})}{d\tilde{\theta}}&=   \frac{\hat{\xi}'(\tilde{\theta})}{\gamma(\theta)} \left( 1 - F\left( \frac{\hat{\xi}(\tilde{\theta})}{\gamma(\theta)}; \tilde{\theta} \right) \right) -  \frac{\hat{\xi}'(\tilde{\theta})}{\gamma(\theta)} \left( 1 - F\left( \frac{\hat{\xi}(\tilde{\theta})}{\gamma};  {\theta} \right) \right)\\&= \frac{\hat{\xi}'(\tilde{\theta})}{\gamma(\theta)} \left(F\left( \frac{\hat{\xi}(\tilde{\theta})}{\gamma(\theta)};  {\theta} \right)-F\left( \frac{\hat{\xi}(\tilde{\theta})}{\gamma(\theta) };  {\tilde\theta} \right)\right)\end{aligned}$$
By Proposition \ref{prop NL theta}(1), we have $\xi(\theta)'<0$. Together with $\partial_\theta F(y,\theta) <0$ (Assumption \ref{ass:FOD}), we have  $$
\frac{\partial J^{\theta}(\hat{l}(\tilde{\theta});\tilde{\theta})}%
{\partial\tilde{\theta}}>(\text{reps.}<)0\text{,\quad for }\tilde{\theta
}<(\text{resp.}>)\theta.
$$
As such,
$$
\theta\in\operatorname{argmax}_{\tilde{\theta}\in\lbrack\theta_{L},\theta
_{H}]}J^{\theta}(\tilde{\theta}).
$$
This completes the proof.
\end{proof}

Theorem~\ref{thm:main}(1) shows that, when the uncertainty concerns customers’ risk types, the optimal risk loading follows the nonlinear pricing rule asdescribed in \eqref{eq:xi_integral}. Moreover,  Theorem~\ref{thm:main}(4) confirms that the optimal menu $(\hat{l}(\theta),\hat{p}(\theta))_{\theta \in [\theta_L, \theta_H]}$ satisfies the truth-telling constraint, meaning that a customer of type $\theta$ achieves no higher utility by selecting a contract $(\hat{l}(\tilde{\theta}),\hat{p}(\tilde{\theta}))$ intended for a different type $\tilde{\theta} \neq \theta$. Consequently, customers' hidden private information regarding their risk types is revealed through the self-selection process.
In the absence of information asymmetry,  the menu of contracts collapses to a single contract,  and is shown in Corollary~\ref{cor: no uncertainty theta}.

\begin{remark}\label{remark:4.3}
The optimal loading factor $\hat\xi(\theta, C^*)$ in Theorem \ref{thm:main} is implicitly defined by the integral equation \eqref{eq:xi_integral} together with the optimization problem $C^* = \operatorname{argmax}_{C\in\mathcal C} H(C)$. Analytical solutions are generally intractable for most loss distributions $Y$, but the numerical procedure can be employed:
 (i) Solve \eqref{eq:xi_integral} using standard ODE solvers (e.g., MATLAB's \texttt{ode45});
(ii) Determine the optimal parameter $C^*$ by maximizing $H(C)$ via one-dimensional numerical optimization (e.g., MATLAB's \texttt{fminbnd}).
The optimal contract terms in \eqref{opt menu theta} are then obtained by substituting $\xi(\theta,C^*)$. This procedure provides an efficient and accurate implementation of the optimal solution using standard computational tools; see Section \ref{sec:5.2} when $Y$ follows the Pareto distribution. 
\end{remark}

The following proposition summarizes the key monotonicity properties of the optimal contract menu in Theorem~\ref{thm:main}.

\begin{proposition}
\label{prop NL theta} 
(1) $\hat{\xi}(\theta)$ is decreasing in $\theta$.

(2) If $\gamma(\theta)$ is increasing in $\theta$, then 
$\hat{l}(\theta)$ is increasing in $\theta$.

(3) If $\gamma(\theta)$ is increasing in $\theta$, then 
$\hat{p}(\theta)$ is increasing in $\theta$.
\end{proposition}

\begin{proof}

(1) Under Assumption~\ref{ass:FOD}, we have $\phi(\theta,\xi)>0$. 
From \eqref{ODE},
\[
\hat{\xi}'(\theta)
=
-(1+\hat{\xi}(\theta))
\phi(\theta,\xi).
\]
Hence $\hat{\xi}(\theta)$ is decreasing in $\theta$.

\medskip

(2) 
If $\gamma(\theta)$ is increasing in $\theta$, then $\hat{d}^\theta(\tilde{\theta})=\hat{\xi}(\theta)/\gamma(\theta)$ is decreasing in $\theta$, and consequently the optimal coverage $\hat l^\theta(\tilde{\theta})$ is increasing in $\theta$.

\medskip

(3) Differentiating $\hat p(\theta)$ gives
\[
\begin{aligned}
\hat p'(\theta)
=&\;\lambda \hat{\xi}'(\theta)\Psi_\theta(\theta,\hat{\xi}(\theta))
+\lambda (1+\hat{\xi}(\theta))\Psi_\theta(\theta,\hat{\xi}(\theta))\\
&-\lambda(1+\hat{\xi}(\theta))
\frac{\hat{\xi}'(\theta)\gamma(\theta)-\hat{\xi}(\theta)\gamma'(\theta)}
{\gamma(\theta)^2}
\left(1-F\!\left(\frac{\hat{\xi}(\theta)}{\gamma(\theta)};\theta\right)\right).
\end{aligned}
\]
Using \eqref{ODE}, the first two terms cancel, which yields
\[
\hat p'(\theta)
=
-\lambda(1+\hat{\xi}(\theta))
\frac{\hat{\xi}'(\theta)\gamma(\theta)-\hat{\xi}(\theta)\gamma'(\theta)}
{\gamma(\theta)^2}
\left(1-F\!\left(\frac{\hat{\xi}(\theta)}{\gamma(\theta)};\theta\right)\right)>0.
\]
Thus,  $\hat p(\theta)$ is increasing in $\theta$.
\end{proof}

Proposition~\ref{prop NL theta}(1) shows that the optimal menu assigns lower risk loadings to higher-risk customers, thereby inducing truthful revelation of types. Importantly, this monotonicity of $\hat{\xi}(\theta)$ does not depend on $\gamma(\theta)$. In fact, the function $\gamma(\theta)$ does not affect the premium schedule. As a result, whether $\gamma$ is constant or depends on $\theta$, the resulting ODE characterizing $\hat{\xi}(\theta)$ remains unchanged, and hence the functional form of the optimal contract is invariant to the specification of $\gamma(\theta)$.
The monotonicity of $\hat{\xi}(\theta)$ contrasts with the case of risk-attitude uncertainty, 
for which the optimal loading remains constant. This difference can be understood as follows. 
Risk-attitude uncertainty affects only the willingness to insure but does not influence expected losses, thereby providing limited informational content for the insurer and not justifying differentiated pricing. By contrast, uncertainty about risk type is economically meaningful, as it directly determines expected claims. This informational relevance gives rise to nonlinear pricing: high-risk customers receive lower marginal loadings, whereas low-risk customers face higher ones. Such a screening mechanism mitigates adverse 
selection and enables the insurer to efficiently extract information about 
risk uncertainty.

Propositions \ref{prop NL theta}(2) and (3) indicate that both the premium 
rate $\hat{p}(\theta)$ and the coverage level $\hat{l}(\theta)$ increase with 
$\theta$, provided that $\gamma(\theta)$ is increasing in $\theta$. This pattern 
reflects the economically intuitive outcome that higher-risk and more risk-averse 
individuals purchase greater coverage, leading to higher total premiums. 
While the  loading per unit of coverage decreases with risk type, 
the larger volume of coverage ensures that the overall premium rises with $\theta$.
If instead $\gamma(\theta)$ decreases in $\theta$, implying that higher-risk individuals are less risk-averse, the demand for coverage reflects a trade-off between increasing risk exposure and decreasing risk aversion. In this case, the monotonicity of the optimal coverage level $\hat{l}(\theta)$ may no longer hold. Nevertheless, the alignment between coverage and total premium is preserved: contracts with higher coverage are associated with higher total premiums, whereas contracts with lower coverage correspond to lower premiums.

\section{Numerical examples}\label{sec:5}
\subsection{Risk-attitude uncertainty}

In this section, we illustrate the theoretical results through numerical examples.   We first consider the case where the insurer faces uncertainty regarding the customer’s risk attitude $\gamma$. 
The main result is presented in  Theorem \ref{thm IN gamma}, which shows that  the optimal loading factor $\hat{\xi}$ is obtained by numerically solving
$$
\hat{\xi} = \arg\max_{\xi \ge 0} \int_{\gamma_L}^{\gamma_H} \left[ \xi \, \mathbb{E}\left( Y - \frac{\xi}{\gamma} \right)_+ - \frac{\gamma_I}{2} \, \mathbb{E}\left( Y - \frac{\xi}{\gamma} \right)_+^2 \right] \, \mathrm dG(\gamma),
$$
after which the optimal coverage and premium are given by
$$
\hat{l}(\gamma) = \left(y - \frac{\hat{\xi}}{\gamma}\right)_+, \quad
\hat{p}(\gamma) = (1+\hat{\xi})\lambda \, \mathbb{E}[\hat{l}(\gamma)].
$$
We  first assume that  $Y \sim \mathrm{exp}(1)$ with customer risk aversion $\gamma$ uniformly distributed over $[\gamma_L, \gamma_H] = [1, 9]$ or $[2, 8]$. The insurer’s risk aversion is fixed at $\gamma_I = 1$. For comparison, we also examine the case without information asymmetry  by fixing $\gamma = 5$, where the results are given by Corrollary \ref{cor: no uncertainty theta}.

\begin{figure}[htbp]
    \centering
    \begin{minipage}[b]{0.45\textwidth}
        \centering
        \includegraphics[width=\textwidth]{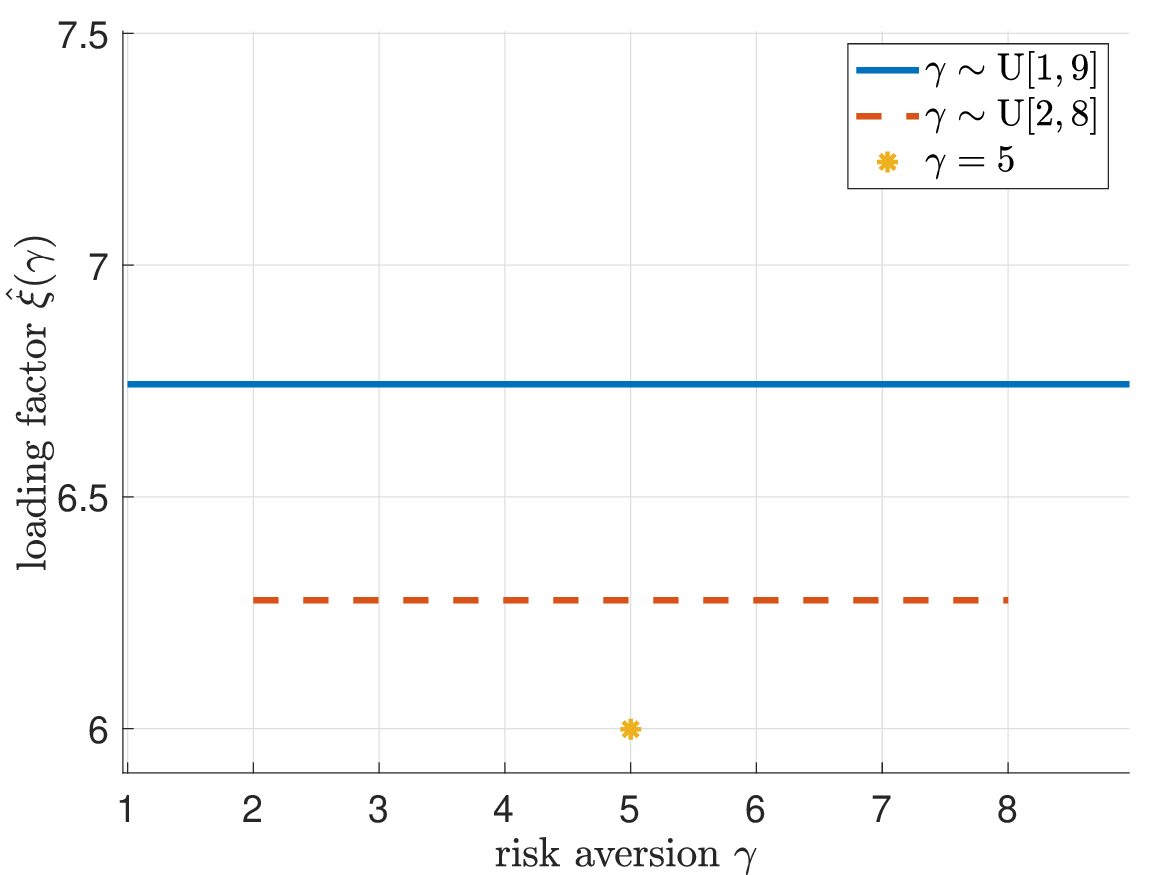}
        \caption*{(a)}
    \end{minipage}
    \hfill
    \begin{minipage}[b]{0.45\textwidth}
        \centering
        \includegraphics[width=\textwidth]{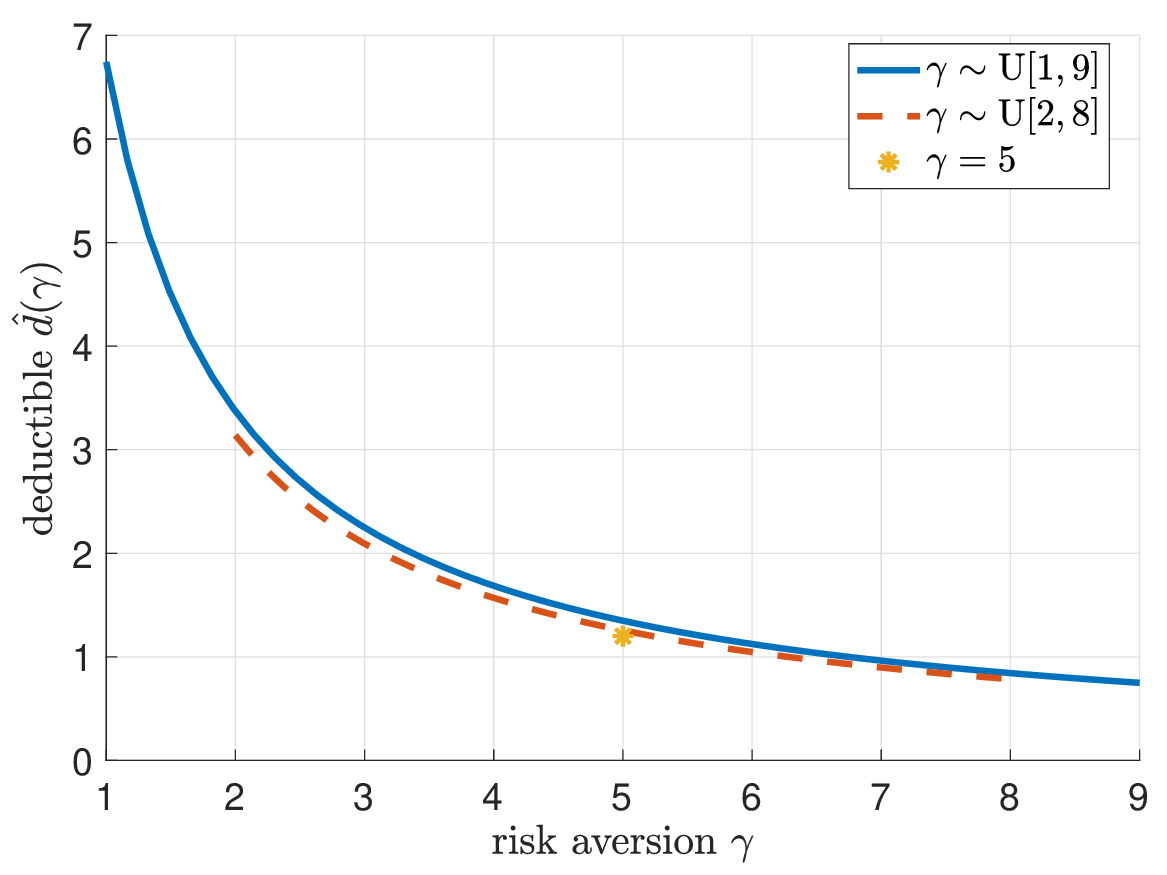}
        \caption*{(b)}
    \end{minipage}
    \hfill

    \caption{Risk loading and deductible functions under risk-attitude uncertainty with $Y \sim \mathrm{exp}(1)$ and  $\gamma$ following a uniform distribution.}
    \label{fig1}
\end{figure}

Figure \ref{fig1}(a) depicts the risk loading function $\hat{\xi}(\gamma)$. As predicted by Proposition \ref{prop opt gamma}(1), the risk loading in the optimal menu remains constant and independent of the level of risk aversion. Accordingly, the curves are flat, indicating that the insurer neither offers discounts nor imposes surcharges through the loading schedule. Moreover, the insurer sets a higher risk loading and offers a wider menu of contracts under greater uncertainty, as illustrated by the case $\gamma \sim \mathrm{U}[1,9]$ compared with $\gamma \sim \mathrm{U}[2,8]$. In addition, the loading under risk-attitude uncertainty is higher than that in the benchmark case without uncertainty, where $\gamma=5$. However, this relationship is not universal. For instance, in the case of uncertainty in risk types in Figure \ref{fig3}, the opposite pattern may arise, where a narrower range of types can lead to a higher level of risk loading.

Figure \ref{fig1}(b) depicts the deductible function $\hat{d}(\gamma)$. Since $\hat{\xi}(\gamma)$ is constant, the deductible $\hat{d}(\gamma)=\hat{\xi}(\gamma)/\gamma$ is therefore decreasing in $\gamma$. This implies that more risk-averse customers optimally retain a smaller share of risk. Furthermore, when the loading price is lower, customers with the same level of risk aversion optimally choose to select a lower deductible, and hence cede a larger portion of risk to the insurer.

Next, we  assume that $Y$ follows a Pareto distribution, i.e., $Y \sim \text{Pareto}(3,3)$. We further assume that $\gamma$ follows a truncated normal distribution supported on the interval $[1,9]$. In particular, we consider the cases $\gamma \sim \mathrm{TruncN}(2,1)$ and $\gamma \sim \mathrm{TruncN}(2,2)$, as well as a fixed value, $\gamma = 2$.  The numerical results are presented in Figure \ref{fig2}. 
   
\begin{figure}[htbp]
    \centering
    \begin{minipage}[b]{0.45\textwidth}
        \centering
        \includegraphics[width=\textwidth]{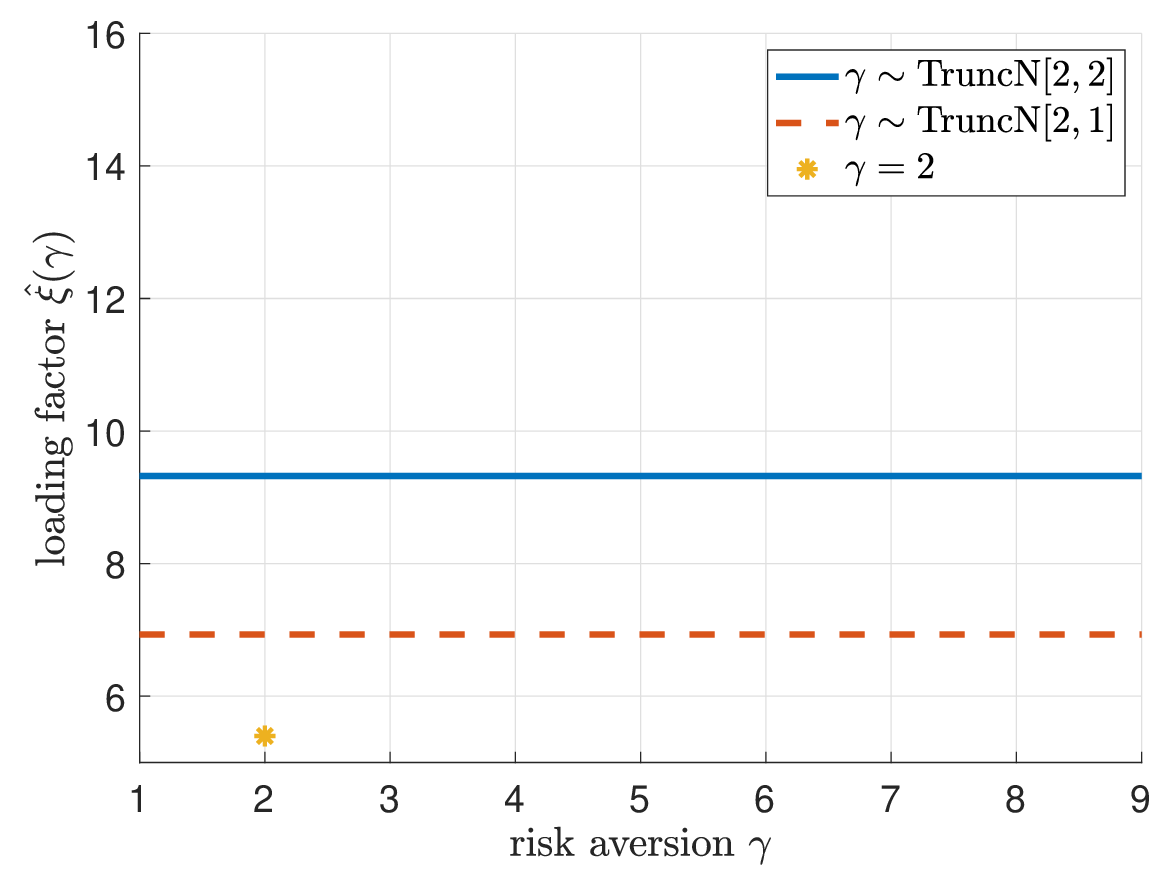}
        \caption*{(a)}
    \end{minipage}
    \hfill
    \begin{minipage}[b]{0.45\textwidth}
        \centering
        \includegraphics[width=\textwidth]{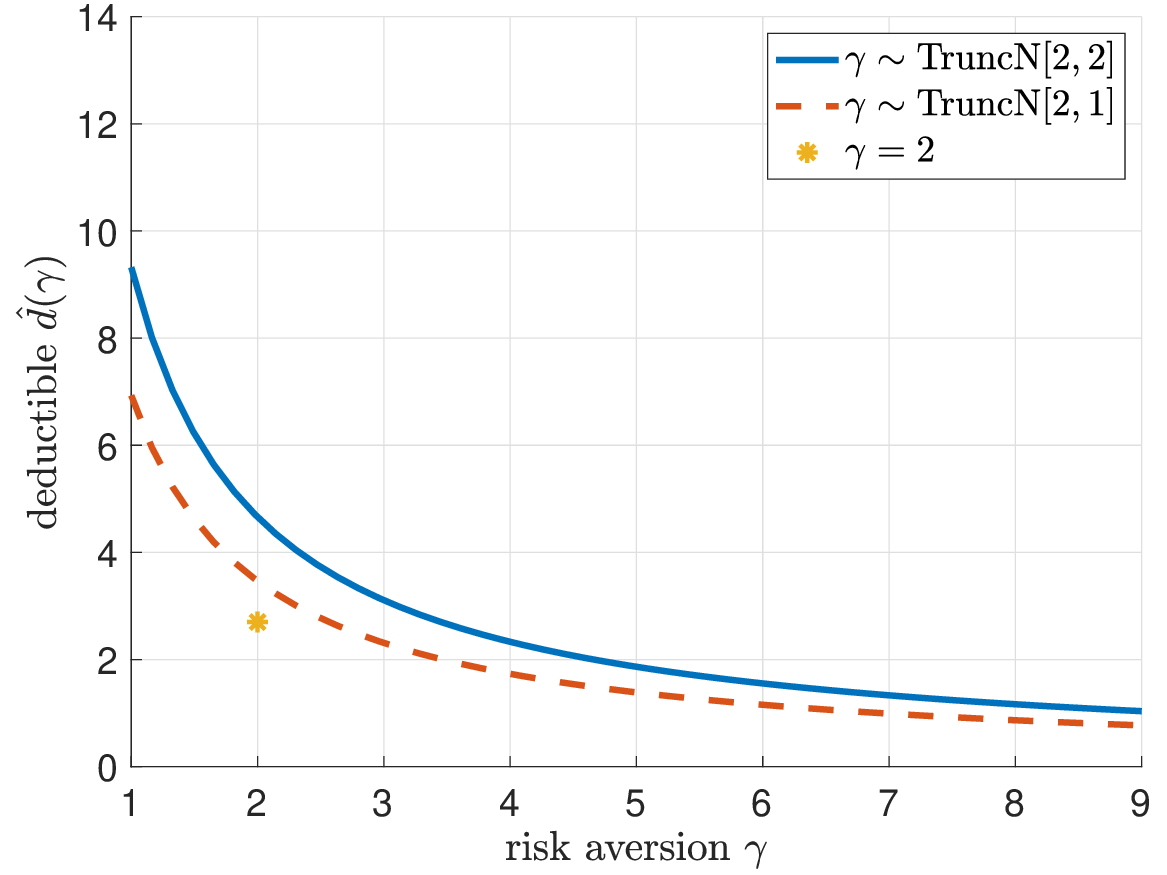}
        \caption*{(b)}
    \end{minipage}
    \hfill

    \caption{Risk loading and deductible functions under risk-attitude uncertainty with $Y \sim \text{Pareto}(3,3)$ and $\gamma$ following a  truncated normal distribution.}
    \label{fig2}
\end{figure}

We can see from Figure \ref{fig2} that the overall patterns remain broadly consistent with those obtained under the uniform distribution. In particular, holding the mean of $\gamma$ fixed, a larger variance in the distribution of risk aversion leads the insurer to set a higher optimal risk loading (Figure \ref{fig2}(a)). As a consequence, customers optimally choose lower insurance coverage (Figure \ref{fig2}(b)).

\subsection{Risk-type uncertainty}\label{sec:5.2}
In this section, we first consider the case where $\gamma=\gamma(\theta)$ is constant and analyze the scenario in which the insurer faces uncertainty about customers’ risk types. We  assume that the loss $Y$ follows an exponential distribution with mean $\theta$, that is, $
Y \sim \exp\!\left(\tfrac{1}{\theta}\right),  \theta \in [\theta_L, \theta_H].$ 
For this distribution, the auxiliary functions in \eqref{eq:Psi} are given by
$$
\Psi(\theta,\xi) = \theta \, e^{-\frac{\xi}{\gamma \theta}}, ~~
\Psi_\theta(\theta,\xi) = \Big(1 + \tfrac{\xi}{\gamma \theta}\Big) e^{-\frac{\xi}{\gamma \theta}}.
$$
Then the differential equation  \eqref{ODE} becomes
$$
\theta \, \hat\xi(\theta)' 
= -\big(1+\hat\xi(\theta)\big)\left(1+\tfrac{\hat\xi(\theta)}{\theta\gamma}\right),
$$
and admits a closed-form  general solution
$$\hat\xi(\theta,C^*) = 
\frac{1}
{1 + \gamma\theta - C^* \theta \exp\!\left(\tfrac{1}{\gamma\theta}\right)}-1,
$$
where the constant $C^*$ is determined by
$
C^* = \operatorname*{arg\,max}_{\underline{C} \leq C \leq \overline{C}} H(C),
$
with 
$$
\begin{aligned}
H(C) &= \int_{\theta_L}^{\theta_H} \int_{0}^{\infty}  
\Bigg[ 
\hat\xi(\theta,C^*) 
\left( y - \frac{\hat\xi(\theta,C^*) }{\gamma}\right)_{+}  - \frac{\gamma_I}{2} \left( y - \frac{\hat\xi(\theta,C^*) }{\gamma}\right)_{+}^{2}
\Bigg] f(y;\theta) \, \mathrm{d}y \, \mathrm{d}G(\theta).
\end{aligned}
$$
The feasible set of constants is bounded by
$$
\underline{C} := \gamma \, e^{-\frac{1}{\gamma \theta_{H}}}, 
\qquad  
\overline{C} := \frac{1+\gamma\theta_{L}}{\theta_{L}} \,
e^{-\frac{1}{\gamma \theta_{L}}},
$$
which guarantees that $\hat{\xi}(\theta) \geq 0$ for all $\theta \in [\theta_L,\theta_H]$.

Finally, the resulting optimal menu of contracts given by  Theorem \ref{thm:main}, characterized by loss coverage $\hat{l}(\theta)$ and premium $\hat{p}(\theta)$, takes the form
$$
\hat{l}(\theta) = \Big( y - \tfrac{\hat{\xi}(\theta)}{\gamma} \Big)_{+}, 
\qquad
\hat{p}(\theta) 
= \big( 1 + \hat{\xi}(\theta) \big)\lambda
\mathbb{E} \left[ \Big( Y - \tfrac{\hat{\xi}(\theta)}{\gamma} \Big)_{+} \right].
$$

\begin{remark}A non-negative solution for $\xi(\theta)$ exists on $[\theta_L, \theta_H]$ if and  only if $\overline{C} \ge \underline{C}$. Equivalently, this imposes a restriction on the parameter interval:
$$
\theta_H \le \frac{\theta_L}{1 - \gamma \theta_L \ln\left(\frac{1+\gamma \theta_L}{\gamma \theta_L}\right)}.
$$
If this condition is violated, i.e., $\overline{C} < \underline{C}$, then no admissible $C$ exists and the ODE has no non-negative solution.
It reflects a situation in which the insurer optimally chooses not to offer any contract, as no agreement can be reached between the insurer and the customers under such uncertainty. \end{remark}

In the following numerical illustration, we assume that the risk type $\theta$ is uniformly distributed over $[\theta_L,\theta_H]=[1,9]$ or $[2,8]$. The risk-aversion parameters of the customer and the insurer are fixed at $\gamma=5$ and $\gamma_I=1$, respectively. For comparison, we also consider the benchmark case without information asymmetry by fixing $\theta=5$, for which the results are characterized in Corollary \ref{cor: no uncertainty theta}.

\begin{figure}[htbp]
    \centering
    \begin{minipage}[b]{0.45\textwidth}
        \centering
        \includegraphics[width=\textwidth]{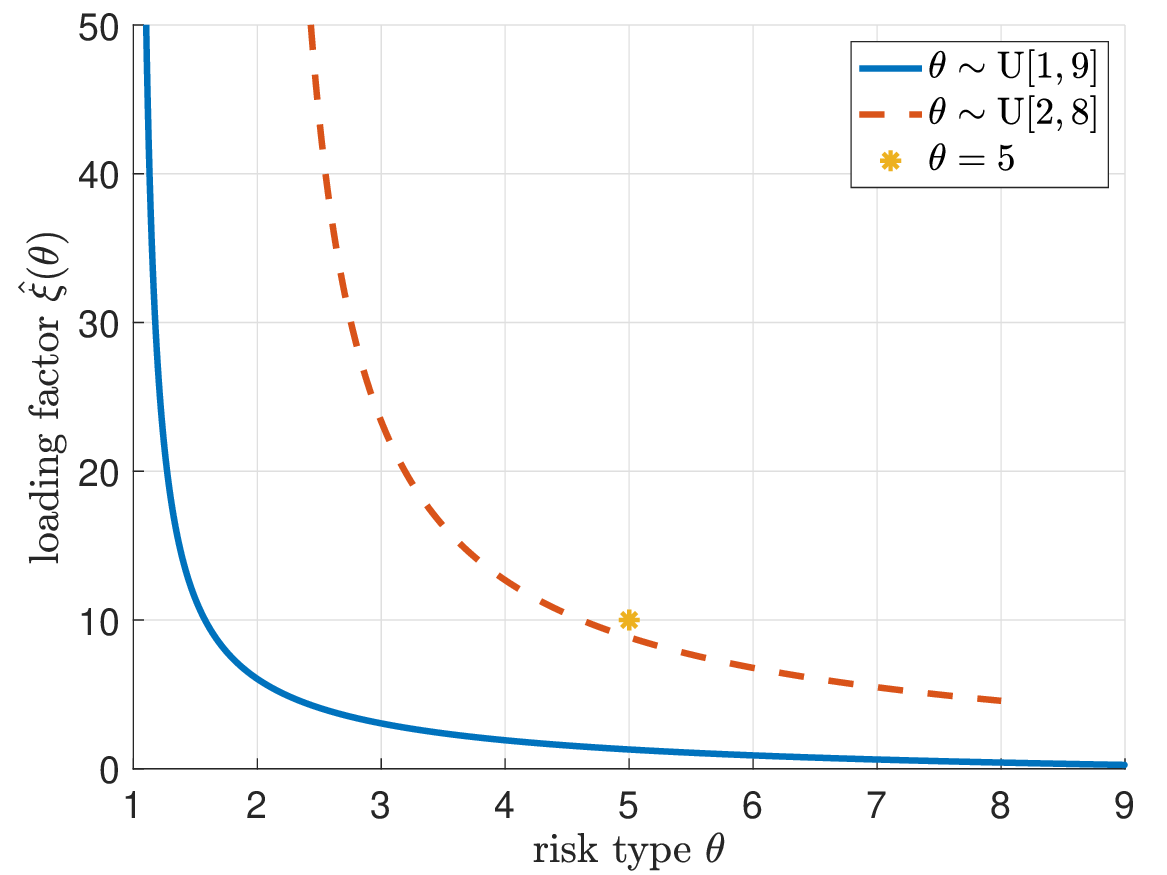}
        \caption*{(a)}
    \end{minipage}
    \hfill
    \begin{minipage}[b]{0.45\textwidth}
        \centering
        \includegraphics[width=\textwidth]{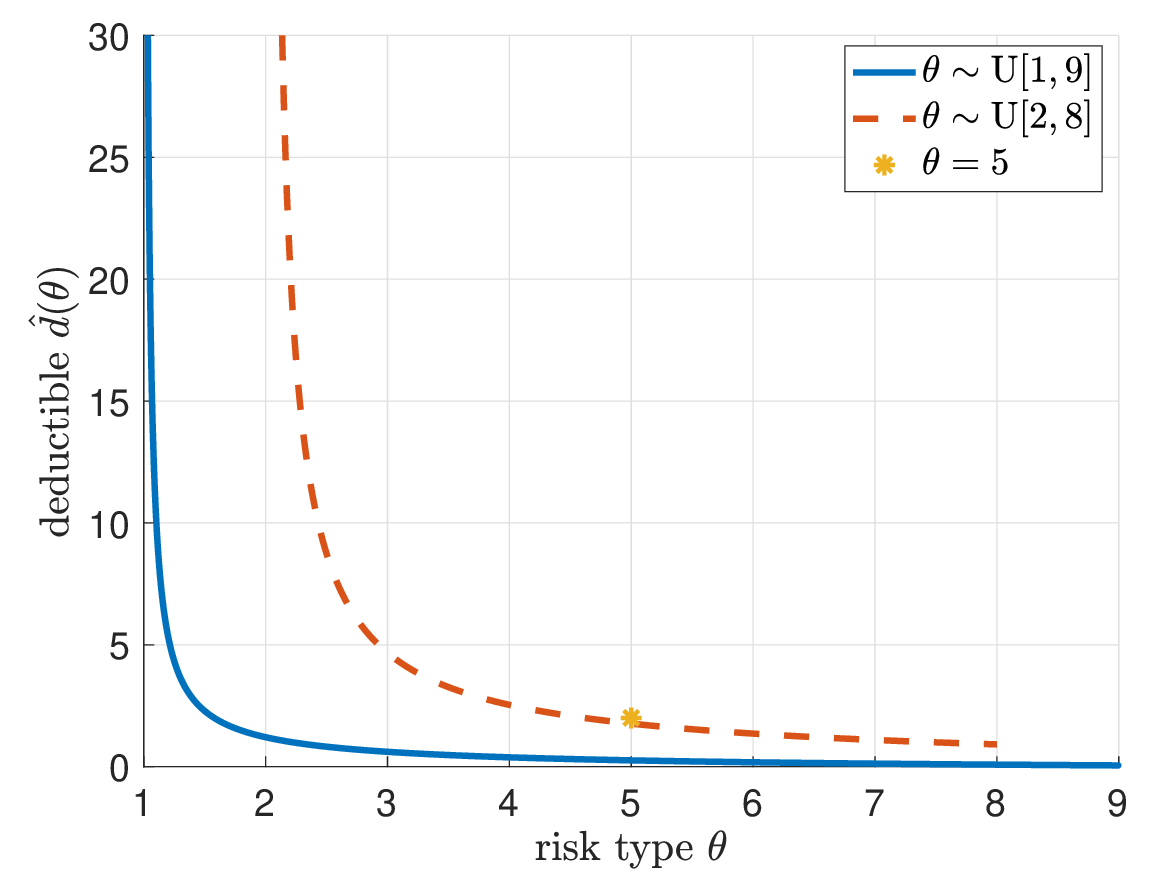}
        \caption*{(b)}
    \end{minipage}

    \caption{Risk loading and deductible functions under risk-type uncertainty with  $Y\sim \mathrm{exp}(\frac{1}{\theta})$ and $\theta$ following a uniform distribution.}
    \label{fig3}
    \end{figure}

\begin{figure}[htbp]
    \centering
    \begin{minipage}[b]{0.45\textwidth}
        \centering
        \includegraphics[width=\textwidth]{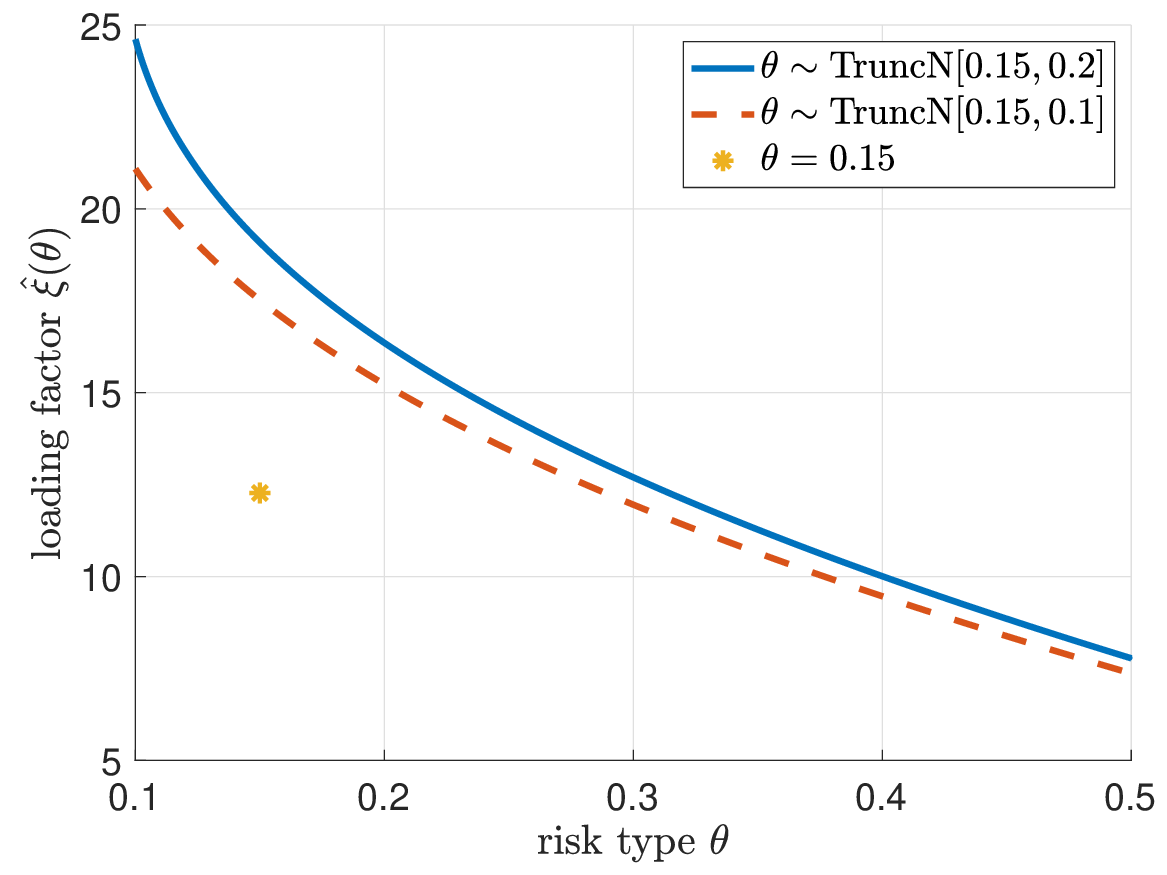}
        \caption*{(a)}
    \end{minipage}
    \hfill
    \begin{minipage}[b]{0.45\textwidth}
        \centering
        \includegraphics[width=\textwidth]{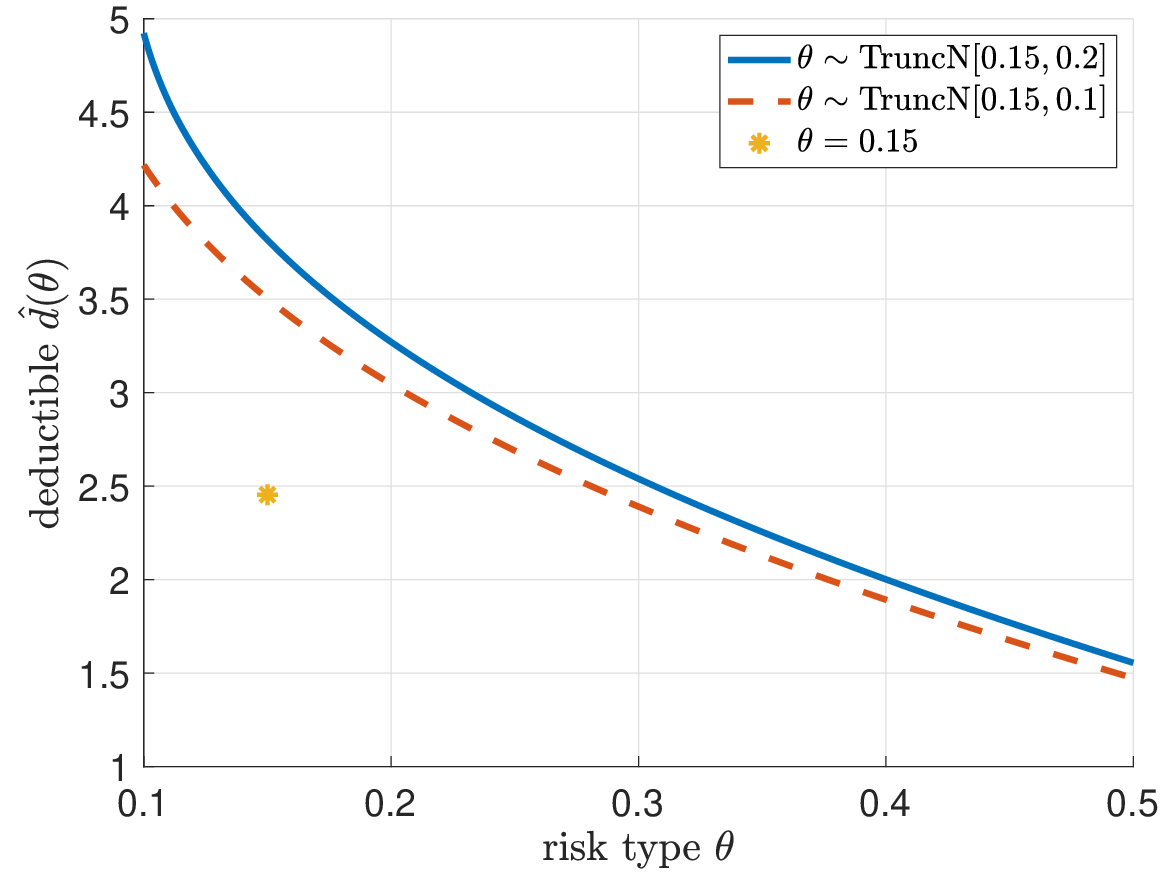}
        \caption*{(b)}
    \end{minipage}
    \caption{  Risk loading and deductible functions under risk-type uncertainty with  $Y\sim \mathrm{Pareto}(\frac{1}{\theta}, 3)$ and $\theta$ following a following a  truncated normal distribution }
    \label{fig4}
    \end{figure}

In Figure \ref{fig3}(a), the risk loading factor $\hat{\xi}(\theta)$ is shown to decrease with the risk type $\theta$, as predicted by Proposition \ref{prop NL theta}(1). This monotonic pattern indicates that higher-risk customers are effectively granted risk-loading discounts. In Figure \ref{fig3}(b), the deductible $\hat d(\theta)=\hat\xi(\theta)/\gamma$ is also decreasing in $\theta$. This is expected since $\hat \xi(\theta)$ is decreasing while $\gamma$ remains constant, implying that customers transfer more risk to the insurer as their inherent risk increases. Moreover, for a given risk type, customers choose greater coverage when the loading price is lower.

Another numerical example is presented in Figure \ref{fig4}. We assume that $Y$ follows a Pareto distribution, i.e.,
$
Y \sim \text{Pareto}\!\left(\tfrac{1}{\theta}, 3\right).
$
In this case, the ODE does not admit a closed-form solution and must be solved numerically, following the procedure described in Remark \ref{remark:4.3}. The parameter $\theta$ is assumed to follow a truncated normal distribution on $[0.1,0.5]$. Specifically, we compare two scenarios with different levels of uncertainty: one with $\theta \sim \text{TruncN}(0.15,0.1)$ and the other with $\theta \sim \text{TruncN}(0.15,0.2)$. For the deterministic benchmark, we fix $\theta=0.15$. The monotonicity of $\hat\xi(\theta)$ and $\hat d(\theta)$ with respect to $\theta$ remains the same as in the previous example where $Y$ follows an exponential distribution.

We also observe in Figure \ref{fig3}(a) that the overall level of the loading increases as the range of risk types becomes narrower, which contrasts with Figure \ref{fig4}(a), where $Y$ follows a Pareto distribution and the loading becomes higher as the variance of $\theta$ increases.
The result that a narrower dispersion of risk types leads to higher premiums is consistent with the classic insights of monopoly insurance models, as in \cite{stiglitz1977monopoly}. It implies that in a more deterministic environment with less uncertainty about customer types, insurers typically set higher prices. By contrast, when uncertainty about types increases, insurers may offer partial discounts in order to induce truthful revelation from customers.
In fact, in our setting the equilibrium arises from a Stackelberg game between the insurer and customers, which lies between the two benchmark cases of pure monopoly and perfect competition. Hence, it is not surprising that the equilibrium loading may either increase or decrease as the level of uncertainty changes. This suggests that the effect of uncertainty  depends crucially on the behavior of the loss distribution.
One possible explanation is related to the tail behavior of losses. When losses follow an exponential distribution, which is light-tailed, the probability of extremely large losses is relatively small. In this case, insurers may provide larger discounts to customers with greater uncertainty in order to encourage truthful reporting. In contrast, when losses follow a heavy-tailed distribution such as the Pareto distribution, large losses occur with higher probability, and insurers respond more conservatively by imposing higher loadings on customers with greater uncertainty.

Finally, we allow the risk-aversion parameter to depend on the risk type, i.e., $\gamma=\gamma(\theta)$. We  maintain $Y \sim \text{Pareto}\!\left(\tfrac{1}{\theta}, 3\right)$, and 
$\theta \sim \mathrm{TruncN}(0.15,0.1)$ on $[0.1,0.5]$. 
We consider three specifications of risk aversion. Under $\gamma(\theta)=50\theta$, risk aversion increases with risk type, implying that higher-risk customers are more risk-averse; in this case, $\gamma(\theta)$ ranges from 5 to 25. 
Under $\gamma(\theta)=0.5/\theta$, risk aversion decreases with risk type, implying that higher-risk customers are less risk-averse, with $\gamma(\theta)$ ranging from 5 to 1. 
For comparison, we also include a benchmark with constant risk aversion $\gamma=5$.
We then examine how the risk loading $\hat{\xi}(\theta)$ and the deductible $\hat d(\theta)$ vary with respect to $\theta$ under these alternative specifications in Figure~\ref{fig5}.
\begin{figure}[htbp]
    \centering
    \begin{minipage}[b]{0.45\textwidth}
        \centering
        \includegraphics[width=\textwidth]{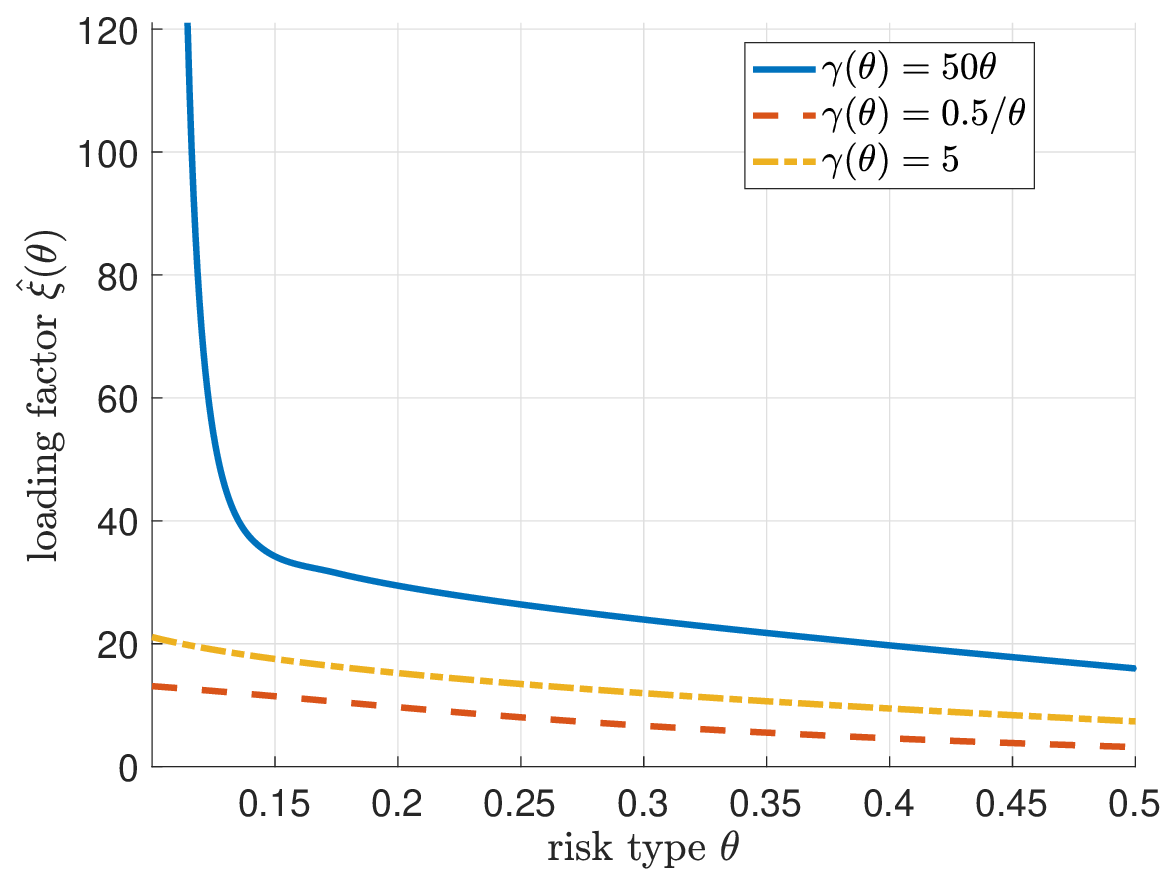}
        \caption*{(a)}
    \end{minipage}
    \hfill
    \begin{minipage}[b]{0.45\textwidth}
        \centering
        \includegraphics[width=\textwidth]{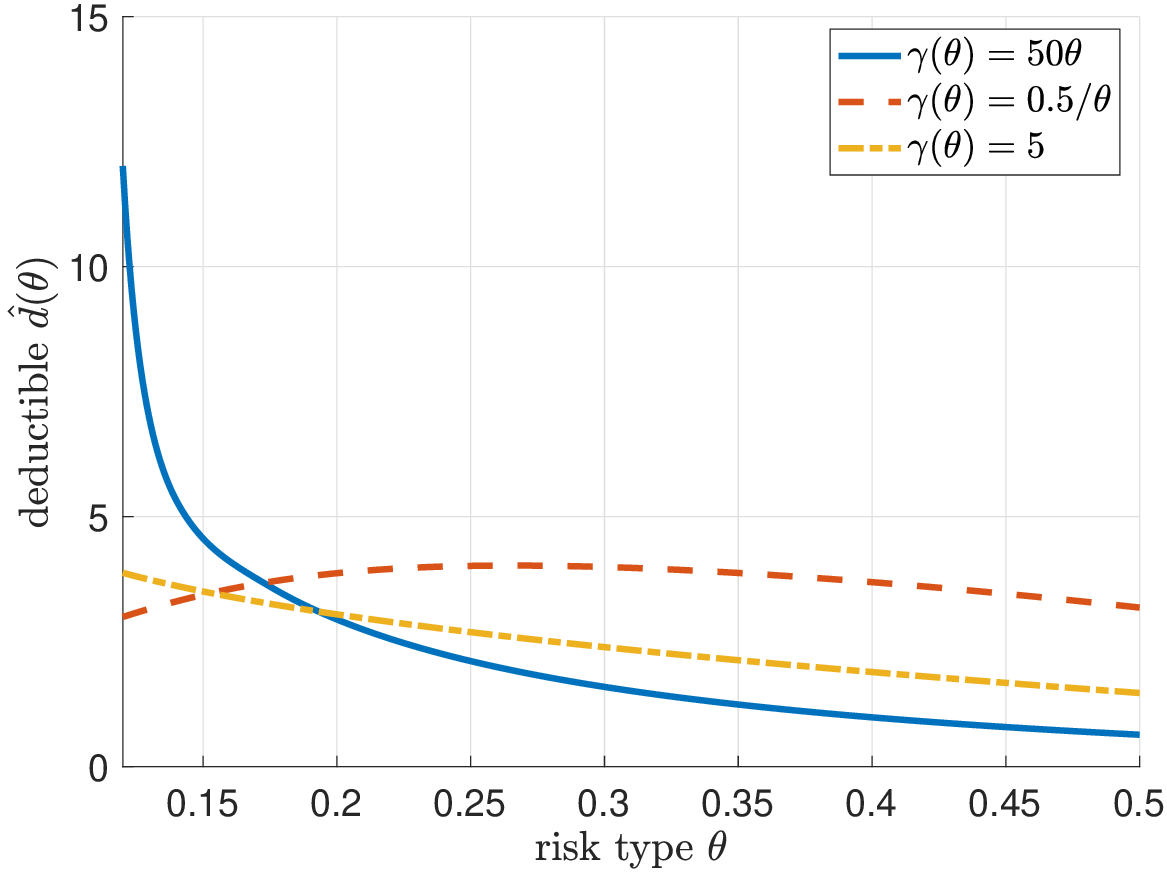}
        \caption*{(b)}
    \end{minipage}


    \caption{Risk loading and deductible functions under risk-type uncertainty with  $Y\sim \mathrm{Pareto}(\frac{1}{\theta}, 3)$ and $\theta\sim\mathrm{TruncN}[0.15, 0.1]$.}
    \label{fig5}
    \end{figure}

Regardless of the functional form of $\gamma(\theta)$, we can see from Figure \ref{fig5}(a) that the risk loading $\hat{\xi}(\theta)$ is decreasing in $\theta$, consistent with Proposition~\ref{prop NL theta}(1). Moreover, when $\gamma(\theta)$ is increasing in $\theta$, the loading curve lies above the constant-$\gamma$ case; when $\gamma(\theta)$ is decreasing in $\theta$, the loading is overall the lowest. This pattern is economically intuitive. A higher level of risk aversion increases customers' willingness to purchase insurance, which enables the insurer to charge a higher loading.

Figure \ref{fig5}(b) shows that 
when $\gamma(\theta)$ is constant or increasing in $\theta$, the deductible $\hat d(\theta)=\hat{\xi}(\theta)/\gamma(\theta)$ decreases with $\theta$ which is proved in Proposition \ref{prop NL theta}(2), indicating that higher-risk customers transfer a larger fraction of risk to the insurer. In contrast, when $\gamma(\theta)$ decreases in $\theta$, the deductible may exhibit a non-monotonic pattern and can first increase and then decrease. Intuitively, when $\theta$ is small, the decline in risk aversion dominates, leading customers to retain more risk. As $\theta$ increases further, the effect of higher risk eventually outweighs the decline in risk aversion, and customers optimally choose to transfer more risk through insurance.
Further, although the increasing-$\gamma(\theta)$ specification leads to the highest premium loading, the associated deductible eventually becomes the smallest due to the strong growth in risk aversion. By contrast, under the decreasing-$\gamma(\theta)$ specification, the premium loading is the lowest, yet the deductible may exceed that in the constant-$\gamma$ case because customers become less risk-averse as the risk type increases.

\section{Conclusion}\label{sec:6}
This study builds on the framework of \citet{han2026optimal} by replacing the variance premium principle with the expected-value premium principle and examining whether the main conclusions remain robust under this alternative pricing rule. Under the variance premium principle, the optimal strategy takes the form of proportional reinsurance, whereas under the expected-value premium principle the optimal contract becomes excess-of-loss. This result is consistent with the full-information findings of \citet{chen2019stochastic}.
From a technical perspective, the main challenge arises because excess-of-loss contracts lead to a first-order condition characterized by an ODE that does not admit a closed-form solution. We address this difficulty by establishing the existence and uniqueness of the optimal contracts using a fixed-point approach.
Our results are broadly consistent with those of \citet{han2026optimal}. When the insurer faces uncertainty about the customer’s risk type, equilibrium contracts exhibit nonlinear pricing with decreasing risk loadings, effectively granting implicit discounts to higher-risk types. By contrast, when uncertainty concerns only risk attitudes, no cross-subsidization arises. 
Several additional observations emerge from our analysis. In particular, we allow the risk-aversion parameter to depend on the risk type. Although the monotonicity of the risk loading is invariant to the specification of the risk-aversion function, the overall level of risk loading is affected by its form. Specifically, the overall risk loading is lowest under decreasing risk aversion, 
highest under increasing risk aversion, and lies in between under constant risk aversion.

Finally, our work opens several directions for future research. A natural extension is to consider a two-dimensional asymmetric information setting in which both risk type and risk attitude are private information, allowing for the analysis of potential interactions and correlations between these characteristics. Moreover, while our model assumes a continuum of contracts, real-world insurance markets typically offer discrete contract menus; studying the efficiency loss arising from such discretization would be both theoretically and empirically valuable. Another promising direction is to generalize the customer's objective function beyond mean–variance preferences, thereby broadening the applicability of the framework to a wider range of risk behaviors.

\subsection*{Acknowledgments}
 Xia Han gratefully acknowledges support from the National Natural Science Foundation of China (Grant Nos. 12301604 and 12371471). Bin Li gratefully acknowledges support from the Natural Sciences and Engineering Research Council of Canada (Grant No. 04338), and support from the National Natural Science Foundation of China (Grant Nos. 12271171 and 12471446).

\bibliography{reference}

@article{stiglitz1977monopoly,
  title={Monopoly, non-linear pricing and imperfect information: The insurance market},
  author={Stiglitz, Joseph E},
  journal={The Review of Economic Studies},
  volume={44},
  number={3},
  pages={407--430},
  year={1977},
  publisher={Wiley-Blackwell}
}

@article{chen2019stochastic,
  title={Stochastic Stackelberg differential reinsurance games under time-inconsistent mean-variance framework},
  author={Chen, Lv and Shen, Yang},
  journal={Insurance: Mathematics and Economics},
  volume={88},
  pages={120--137},
  year={2019},
  publisher={Elsevier}
}

@article{rothschild1976equilibrium, 
title = {Equilibrium in Competitive Insurance Markets: An Essay on the Economics of Imperfect Information},
author = {Rothschild, Michael and Stiglitz, Joseph},
year = {1976},
journal = {The Quarterly Journal of Economics},
volume = {90},
number = {4},
pages = {629-649}
}

@article{chade2012optimal,
  title={Optimal insurance with adverse selection},
  author={Chade, Hector and Schlee, Edward},
  journal={Theoretical Economics},
  volume={7},
  number={3},
  pages={571--607},
  year={2012},
  publisher={Wiley Online Library}
}

@article{gershkov2023optimal,
  title={Optimal Insurance: Dual Utility, Random Losses, and Adverse Selection},
  author={Gershkov, Alex and Moldovanu, Benny and Strack, Philipp and Zhang, Mengxi},
  journal={American Economic Review},
  volume={113},
  number={10},
  pages={2581--2614},
  year={2023},
  publisher={American Economic Association 2014 Broadway, Suite 305, Nashville, TN 37203}
}

@article{hendren2013private,
  title={Private information and insurance rejections},
  author={Hendren, Nathaniel},
  journal={Econometrica},
  volume={81},
  number={5},
  pages={1713--1762},
  year={2013},
  publisher={Wiley Online Library}
}

@article{bjork2010general,
  title={A general theory of Markovian time inconsistent stochastic control problems},
  author={Bj\"ork, Tomas and Murgoci, Agatha},
  journal={Available at SSRN 1694759},
  year={2010}
}

@article{garcia2021information,
  title={Information design in competitive insurance markets},
  author={Garcia, Daniel and Tsur, Matan},
  journal={Journal of Economic Theory},
  volume={191},
  pages={105160},
  year={2021},
  publisher={Elsevier}
}

@article{farinha2023risk,
  title={Risk classification in insurance markets with risk and preference heterogeneity},
  author={Farinha Luz, Vitor and Gottardi, Piero and Moreira, Humberto},
  journal={Review of Economic Studies},
  volume={90},
  number={6},
  pages={3022--3082},
  year={2023},
  publisher={Oxford University Press US}
}

@article{li2016alpha,
  title={Alpha-robust mean-variance reinsurance-investment strategy},
  author={Li, Bin and Li, Danping and Xiong, Dewen},
  journal={Journal of Economic Dynamics and Control},
  volume={70},
  pages={101--123},
  year={2016},
  publisher={Elsevier}
}

@article{gollier2014optimal,
  title={Optimal insurance design of ambiguous risks},
  author={Gollier, Christian},
  journal={Economic Theory},
  volume={57},
  pages={555--576},
  year={2014},
  publisher={Springer}
}

@article{guan2022equilibrium,
  title={Equilibrium investment and reinsurance strategies under smooth ambiguity with a general second-order distribution},
  author={Guan, Guohui and Li, Bin},
  journal={Journal of Economic Dynamics and Control},
  volume={143},
  pages={104515},
  year={2022},
  publisher={Elsevier}
}

@article{wambach2000introducing,
  title={Introducing heterogeneity in the Rothschild-Stiglitz model},
  author={Wambach, Achim},
  journal={Journal of Risk and Insurance},
  pages={579--591},
  year={2000},
  publisher={JSTOR}
}

@article{cheung2025optimal,
  title={Optimal design of reinsurance contracts under adverse selection with a continuum of types},
  author={Cheung, Ka Chun and Yam, Sheung Chi Phillip and Yuen, Fei Lung and Zhang, Yiying},
  journal={arXiv preprint arXiv:2504.17468},
  year={2025}
}

@article{boonen2021optimal,
  title={Optimal reinsurance design with distortion risk measures and asymmetric information},
  author={Boonen, Tim J and Zhang, Yiying},
  journal={ASTIN Bulletin: The Journal of the IAA},
  volume={51},
  number={2},
  pages={607--629},
  year={2021},
  publisher={Cambridge University Press}
}

@article{cheung2020concave,
  title={Concave distortion risk minimizing reinsurance design under adverse selection},
  author={Cheung, Ka Chun and Yam, Sheung Chi Phillip and Yuen, Fei Lung and Zhang, Yiying},
  journal={Insurance: Mathematics and Economics},
  volume={91},
  pages={155--165},
  year={2020},
  publisher={Elsevier}
}

@article{yaari1987dual,
  title={The dual theory of choice under risk},
  author={Yaari, Menahem E.},
  journal={Econometrica},
   volume={55},
  number={1},
  pages={95--115},
  year={1987},
  publisher={JSTOR}
}

@article{han2026optimal,
  title={Optimal Insurance with Information Asymmetry: Nonlinear and Linear Pricing},
  author={Han, Xia and Li, Bin and Luo, Yao},
  journal={Journal of Economic Dynamics and Control},
  pages={105265},
    volume={184},
  year={2026},
  publisher={Elsevier}
}

@article{cheung2019reinsurance,
  title={Reinsurance contract design with adverse selection},
  author={Cheung, Ka Chun and Yam, Sheung Chi Phillip and Yuen, Fei Lung},
  journal={Scandinavian Actuarial Journal},
  volume={2019},
  number={9},
  pages={784--798},
  year={2019},
  publisher={Taylor \& Francis}
}

@article{liang2022revisiting,
  title={Revisiting the optimal insurance design under adverse selection: Distortion risk measures and tail-risk overestimation},
  author={Liang, Zhihang and Zou, Jushen and Jiang, Wenjun},
  journal={Insurance: Mathematics and Economics},
  volume={104},
  pages={200--221},
  year={2022},
  publisher={Elsevier}
}

@article{ghossoub2025optimal,
  title={Optimal Insurance in a Monopoly: Dual Utilities with Hidden Risk Attitudes},
  author={Ghossoub, Mario and Li, Bin and Shi, Benxuan},
  journal={arXiv preprint arXiv:2504.01095},
  year={2025}
}
\end{document}